\documentclass[journal]{IEEEtran}

\usepackage[dvipdf]{graphicx,color}
\usepackage{epsfig}
\usepackage{graphpap}
\usepackage{amssymb}
\usepackage[cmex10]{amsmath}
\usepackage{amsfonts}
\usepackage{amsmath}
\usepackage{multirow}
\usepackage{algorithm}
\usepackage{algorithmicx}
\usepackage{algpseudocode}
\usepackage{tabularx}
\usepackage{color}
\usepackage{xcolor}
\usepackage{varwidth}
\usepackage{multicol}
\usepackage{subfigure}
\usepackage{enumerate}
\usepackage{amsthm}
\usepackage{indentfirst}
\usepackage{exscale}
\usepackage[font=footnotesize]{caption}
\usepackage{stfloats}
\usepackage{float}

\makeatletter
\newcommand{\Statey}{\Statex\hspace*{\ALG@thistlm}}
\makeatother
\makeatletter
\newcommand{\multiline}[1]{%
  \begin{tabularx}{\dimexpr\linewidth-\ALG@thistlm}[t]{@{}X@{}}
    #1
  \end{tabularx}
}
\makeatother

\newtheorem{theorem}{Theorem}
\newtheorem{lemma}{Lemma}

\begin{document}

%
\title{UAV Trajectory, User Association and Power Control for Multi-UAV Enabled Energy Harvesting Communications: Offline Design and Online Reinforcement Learning}

\author{\IEEEauthorblockN{Chien-Wei Fu, {\textit{Student Member}}, {\textit{IEEE}}}, Meng-Lin Ku, {\textit{Senior Member}}, {\textit{IEEE}}, Yu-Jia Chen, {\textit{Senior Member}}, {\textit{IEEE}}, and Tony Q. S. Quek, {\textit{Fellow}}, {\textit{IEEE}} 

\thanks{Chien-Wei Fu, Meng-Lin Ku and Yu-Jia Chen are with the Department of Communication Engineering, National Central University, Zhongli 32001, Taiwan (E-mail: fuchienwei8666@gmail.com, mlku@ce.ncu.edu.tw, yjchen@ce.ncu.edu.tw). Tony Q. S. Quek is with the Information Systems Technology and Design (ISTD) Pillar, Singapore University of Technology and Design, Singapore 487372 (e-mail: tonyquek@sutd.edu.sg).

\textit{Corresponding author: Meng-Lin Ku}}}

\maketitle

\begin{abstract}
In this paper, we consider multiple solar-powered wireless nodes which utilize the harvested solar energy to transmit collected data to multiple unmanned aerial vehicles (UAVs) in the uplink. In this context, we jointly design UAV flight trajectories, UAV-node communication associations, and uplink power control to effectively utilize the harvested energy and manage co-channel interference within a finite time horizon. To ensure the fairness of wireless nodes, the design goal is to maximize the worst user rate. The joint design problem is highly non-convex and requires causal (future) knowledge of the instantaneous energy state information (ESI) and channel state information (CSI), which are difficult to predict in reality. To overcome these challenges, we propose an offline method based on convex optimization that only utilizes the average ESI and CSI. The problem is solved by three convex subproblems with successive convex approximation (SCA) and alternative optimization. We further design an online convex-assisted reinforcement learning (CARL) method to improve the system performance based on real-time environmental information. An idea of multi-UAV regulated flight corridors, based on the optimal offline UAV trajectories, is proposed to avoid unnecessary flight exploration by UAVs and enables us to improve the learning efficiency and system performance, as compared with the conventional reinforcement learning (RL) method. Computer simulations are used to verify the effectiveness of the proposed methods. The proposed CARL method provides $25\%$ and $12\%$ improvement on the worst user rate over the offline and conventional RL methods.
\end{abstract}


\section{Introduction}
In the era of Internet of Things (IoT), many wireless communication nodes are deployed over wide areas for applications toward sustainable development, e.g., environmental monitoring. Traditional ways to powering up electrical devices by connecting the power grids or non-rechargeable batteries with frequent replacement incur high deployment and maintenance cost. In recent years, energy harvesting (EH) technology has been recognized as an effective means to meet the energy requirement of electrical devices and prolong the lifetime of wireless networks. Due to its ubiquitous abundance, solar power is one of the most preferable EH sources and is capable of providing nearly-permanent power supply. However, renewable energy sources are often affected by the environment resulting in stochastic EH, and therefore the power control design of communication nodes becomes an important issue for self-sustaining wireless services {\cite{J. Lu13}}.

More recently, unmanned aerial vehicle (UAV) communications have received tremendous attention in various applications such as data collection {\cite{C. You20}}, wireless power transfer {\cite{R. Chen20}\nocite{M. A. Ali22}\nocite{Y. Che21}\nocite{Y. Lai18}-\cite{S. A. Hoseini20}}, relaying {\cite{Y. Hu20}}, and mobile edge computing {\cite{X. Xu21}}, due to the potential to improve the transmission quality with the significant advantages of high mobility, flexible deployment, and low cost {\cite{Y. Zeng19}}. In contrast to the wireless networks with terrestrial infrastructures, the UAV can be flexibly deployed to collect sensing data from widely distributed nodes by dynamically adjusting its location. This can dramatically improve the energy efficiency and facilitate the successful deployment of EH communications. Multi-UAV can form multiple mobile base stations to improve the network performance, but interference management becomes a critical challenge when multiple nodes communicate with multiple UAVs concurrently. Since the flight of UAVs is limited by the battery power, it is crucial to investigate wireless resource allocation in UAV-enabled communications. 

When multi-UAV-enabled communications are designed for serving multi-EH wireless nodes (WNs), several design challenges need to be carefully addressed. First, the communication association between the UAVs and the WNs is crucial for mitigating the interference from the concurrently served multi-EH WNs to the multi-UAVs, either through appropriate scheduling of alternate transmissions or by selecting WNs with the least interference impact to other users. The communication association is especially important for EH WNs with limited energy. Second, the power control of EH WNs is subject to energy causality constraints, where the energy consumed by power control at any time should not exceed the total amount of energy collected from the past to the present, and battery storage constraints, where energy charging at any time is limited to the maximum battery capacity {\cite{K. Singh18}}. Third, the UAV flight has initial and final flight position constraints, maximum flight speed constraints, flight altitude constraints, and safety distances between multiple UAVs for collision avoidance {\cite{G. Zhang19}-\cite{W. Mei19}}. In particular, the flight paths of the multiple UAVs will interact with each other, which has a huge impact on the system performance. In this regard, the joint design of the transmit power control of EH WNs, the communication association between UAVs and EH WNs, and the multi-UAV flight trajectories is a problem worthy of study, as they are closely related to each other. The joint design over the time horizon typically relies on the channel state information (CSI) and {energy state information (ESI)\footnote{The ESI refers to the amount of harvested solar energy.}} from the current time to the future time. Unfortunately, it is very difficult to predict and obtain accurate non-causal CSI and ESI in reality due to the random nature of harvested energy and the susceptibility of wireless channels. While some dynamic programming methods, such as reinforcement learning (RL), can learn an optimal policy based only on the current system states by repeatedly interacting with the environment, such methods suffer from the curse of dimensionality when the cardinality of the system state and action space is high. Therefore, we are motivated to design offline methods which only require the statistical/average or instantaneous CSI and ESI, called \textit{offline methods}, and the developed offline strategies can provide some useful insights to improve RL-based \textit{online methods}.

\subsection{Literature Survey}
Many researchers have contributed greatly to the topics of UAV deployment and resource management for UAV communications. There are some early works discussing the deployment of a single UAV. In \cite{Y. Sun19}, successive convex approximation (SCA) is used to optimize the UAV trajectory, transmit power, and subcarrier allocation for throughput maximization. In \cite{C. You20}, a maximum-minimum data collection rate problem is solved for UAV wireless sensor networks in urban areas, where the three-dimensional UAV trajectory and transmission scheduling of sensors are jointly optimized by convex approximation. In \cite{S. Salehi20}, the UAV trajectory is optimized via Q-learning to improve energy efficiency while ensuring QoS of ground users. For multiple UAVs, the interference problem becomes a key issue dominating wireless communication performance and requires new solutions to exploit its inherent spatial degree of freedom. The literature \cite{S. Zhang19} maximizes the uplink transmission rate by designing multi-UAV cooperative transmission, sub-channel allocation and UAV speed. In \cite{Q. Wang19}, UAV flight trajectories are investigated through a dueling deep Q-network to maximize downlink channel capacity under the line-of-sight (LOS) channel probability and user coverage constraint. The joint design of communication scheduling, power control, and UAV trajectory is proposed to maximize the minimum throughput in \cite{Q. Wu18} and to minimize the weighted sum of aerial and ground costs in \cite{C. Zhan20} via alternative optimization and SCA, while the authors in \cite{C. Shen20} focus on the UAV flight speed, altitude, and power control problems. To simplify the design, most studies, such as \cite{R. Chen20}\cite{Q. Wu18}, only consider LOS channels and ignore small-scale fading. To make the research more realistic for UAV applications, the LOS/NLOS channels are emphasized in \cite{C. You20}\cite{S. Zhang19}, and the effect of small-scale fading is evaluated in \cite{Y. Sun19}\cite{Q. Wang19}\cite{C. Zhan20}\cite{C. Shen20}.

Another line of work in the existing literature is to exploit EH in UAV communications for prolonging the lifetime of wireless devices. In \cite{R. Chen20}, UAVs with radio-frequency (RF) EH capability are utilized to assist mobile edge computing, which can provide users with edge computing services and energy through wireless charging. In \cite{M. A. Ali22}, a UAV with EH is considered for uninterrupted service, where harvesting and charging time, flight trajectory and speed, and UAV's transmit power allocation are jointly optimized by block coordinate descent and SCA methods. A dynamic fly-hover-transmission scheme is studied in \cite{Y. Che21}\cite{Y. Lai18} for UAV-assisted wireless energy and information transfer in cognitive radio networks, where a constrained Markov decision process (MDP) problem is cast to design the UAV transmission and trajectory based on causal system information for throughput maximization, while \cite{Y. Lai18} proposes an efficient suboptimal but low-complexity transmission policy. In \cite{S. A. Hoseini20}, a Q-learning method is used to find the flight policy for UAVs as power stations to maximize mission duration. However, the main drawback of wireless charging is the low charging efficiency due to channel path loss. Also, in some hazardous areas, RF signals may not be readily available or dense enough. As an alternative, EH nodes with the self-sustainability from renewable energy are fascinating in UAV communications and enable UAVs to focus more on data collection and improve system performance \cite{J. Bowman20}\cite{Q. V. Do22}. In \cite{J. Bowman20}, a single UAV with a fixed flight path is dispatched to collect data from multiple sensing nodes equipped with solar cells, which can extend network lifetime without human intervention. In \cite{Q. V. Do22}, the UAV placement and resource allocation are investigated through deep learning in the renewable energy paradigm.

\subsection{Contributions}

Although there has been prior work on the study of UAV trajectory and resource allocation for single-UAV scenarios with solar EH nodes \cite{Q. V. Do22}, multi-UAV scenarios with solar EH nodes have not been investigated in the literature. {In addition, most existing multi-UAV design frameworks only consider the LOS channels and neglect the small-scale fading effect}. Although LOS/NLOS channels are considered in \cite{C. You20}\cite{S. Zhang19} and small-scale fading is considered in \cite{Y. Sun19}\cite{Q. Wang19}\cite{C. Zhan20}\cite{C. Shen20}, the multi-UAV communications with EH nodes remain unexplored. Besides, in the context of power control and UAV trajectory, the MDP approaches {\cite{Y. Che21}\cite{Y. Lai18}} necessitate the state transition probability in advance, while the convex approaches typically achieve better performance {\cite{C. You20}\cite{M. A. Ali22}\cite{Y. Sun19}\cite{S. Zhang19}\cite{Q. Wu18}\nocite{C. Zhan20}-\cite{C. Shen20}} under the assumption that the future CSI and ESI are perfectly known. While the RL can maximize long-term utility without knowing the exact stochastic transition dynamics as in {\cite{S. A. Hoseini20}\cite{S. Salehi20}\cite{Q. Wang19}\cite{Q. V. Do22}}, it notably suffers from the curse of dimensionality for systems with large-scale state and action sets, which can lead to ineffective convergence of Q-values. To fill these gaps, this paper presents a framework for jointly designing multi-UAV flight trajectories, communication association between UAVs and nodes, and uplink transmit power control of multi-UAV communication networks with solar-powered ground nodes. The main contributions of this paper are stated as follows.

$\bullet$ To the best of our knowledge, this is the first work to investigate a multi-UAV communication network with multiple solar-powered ground nodes, with a focus on the resource management strategies and multi-UAV flight paths planning to minimize the worst-case user rate. 

$\bullet$ A real solar power harvesting dataset and a composite channel model, including LOS/NLOS and small-scale fading, are considered in this design. The joint design problem is non-convex. Based on a series of SCAs, we propose an offline method for jointly optimizing multi-UAV flight trajectories, UAV-node communication associations and transmit power control which only requires average CSI and ESI. The offline design provides valuable insight into the flight path planning of multiple UAVs, especially when future instantaneous CSI and ESI are unpredictable.

$\bullet$ The online RL is then proposed by further taking the current instantaneous CSI and ESI into account. With the help of multi-UAV trajectories in the offline design, a new concept of UAV flight corridor is presented in the online RL design, called convex-assisted RL (CARL), to effectively guide the executed actions and state exploration. By arranging flight corridors ahead together with the offline UAV trajectories, the proposed online method can not only avoid exploring the state space randomly but also effectively respond to the current instantaneous CSI and ESI for performance improvement. 

$\bullet$ Computer simulations are performed to demonstrate the effectiveness of the proposed offline and online methods. The proposed offline method performs better than other baseline schemes with fixed UAV paths or fixed uplink transmit power. Furthermore, the proposed CARL method significantly improves the performance, compared to the conventional RL and the proposed offline methods.

{The rest of this paper is organized as follows. In Section \ref{System_model}, we present the system model and the joint design problem of the UAV flight trajectory, communication association, and power control of solar EH nodes, along with the MDP formulation. In Section \ref{offline}, an offline convex optimization method that utilizes only the average CSI and ESI is proposed. An online CARL method based on the proposed offline design is given in Section \ref{Online}. Simulation results are provided in Section \ref{Simulation}, and Section \ref{Conclusion} concludes the paper.}

\section{System Model And Problem Formulation}\label{System_model}
{In this section, we first present the system model for multi-UAV communications with multiple solar-powered ground nodes, in which the ground wireless nodes harvest energy from solar and transmit data to the UAVs in the uplink channels over the same frequency band. Afterwards, the joint design problem of multi-UAV flight trajectory, communication association between UAVs and ground nodes, and power control is formulated for multiple UAVs serving multiple uplink solar-powered nodes.}
\subsection{System Model}
\begin{figure}[h]
\centering
\includegraphics[width=0.48\textwidth]{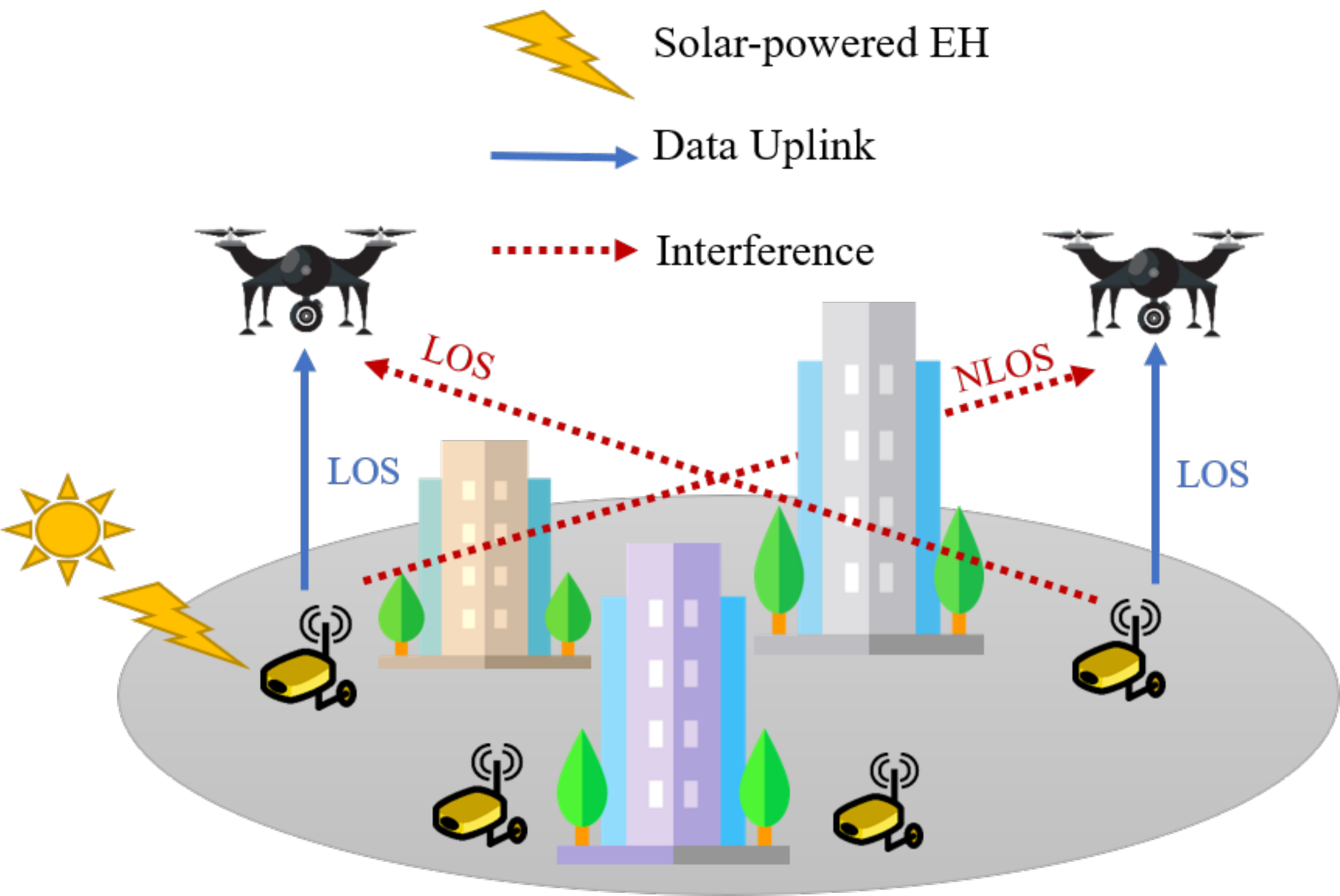}
\caption{Multi-UAV communication networks with solar-powered ground nodes for uplink transmissions ($M$=2 and $K$=4).}
\label{Fig_1}
\end{figure}
Fig. \ref{Fig_1} shows the multi-UAV communication networks with solar-powered ground nodes for uplink transmissions. Consider a service area consisting of a group of $K$ wireless nodes (WNs) which utilize the harvested solar power for data transmission. A group of $M$ UAVs flies over the area to collect the data from the $K$ wireless nodes. Both the UAVs and the WNs are only equipped with a single antenna. A time-slotted model is adopted, and we assume that the entire task period $T_s$ is divided into $N$ time slots, where there are $N+1$ discrete time instants ($n=0,1,\ldots,N$) and the time interval $\delta_D$ is defined as
\begin{align}\label{eq1}
&\delta_D=\frac{T_s}{N} \,.
\end{align}We assume that the positions of the WNs are unchanged, and the horizontal two-dimensional coordinate of the $k$th WN is given by
\begin{align}
&\textbf g_k=[{\bar{x}}_k,{\bar{y}}_k]^T \in \mathbb{R}^2 , \forall k\in\{1,\ldots,K\}\label{eq2} \,,
\end{align}where $\left[\cdot\right]^T$ takes the vector transpose. It is assumed that the UAVs fly at a fixed altitude $H$, and the horizontal two-dimensional coordinate of the $m$th UAV at the time instant $n$ is given as
\begin{align}\label{eq3}
&\textbf q_m [n]=[x_m [n],y_m [n]]^T \in \mathbb{R}^{2}, \nonumber \\
& \qquad\qquad \forall m\in\{1,\ldots,M\}, \forall n \in\{ 0,\ldots,N\} \,.
\end{align}The trajectory coordinate of the UAVs is subject to the following constraints:
\begin{align}
\textbf q_m [0]=\textbf q_m [N]=\textbf q_m^{ini} ,  \forall m\in\{1,\ldots,M\} \,;\label{eq4}
\end{align}
\begin{align}
&\qquad\qquad\quad\frac{1}{\delta_D}\left\|\textbf q_m[n+1]-\textbf q_m[n]\right\|_2 \leq V_{max} , \nonumber\\
&\qquad\qquad \forall m\in\{1,\ldots,M\}, \forall n \in\{ 0,\ldots,N-1\} \,; \label{eq5}
\end{align}
\begin{align}
&\left\|\textbf q_m[n]-\textbf q_j[n]\right\|_2\geq D_{min} ,\nonumber\\
&\;\forall n \in\{ 1,\ldots,N-1\},\forall m,j\in\{1,\ldots,M\},m\neq j \,, \label{eq6}
\end{align}
where $\textbf q_m^{ini}$ is the initial position of the UAVs, $V_{\max}$ is the maximum flight speed of the UAVs, and $D_{\min}$ is the minimum safe distance for any two UAVs. We assume that the UAVs are equipped with finite batteries, and the constraint (\ref{eq4}) stipulates that each UAV is required to fly back to the initial point $\textbf q_m^{ini}$ at the end of the task period. The constraint (\ref{eq5}) indicates that the flight speed of each UAV is limited to the maximum flight speed $V_{max}$. Moreover, the constraint (\ref{eq6}) represents that the minimum safe distance $D_{min}$ must be maintained between any two UAVs for avoiding collision.

For the channel model, both the LOS and NLOS channels are taken into consideration in the large-scale fading, and the path loss (in decibel) between the $m$th UAV and the $k$th WN at the time instant $n$ can be expressed as \cite{Al-Hourani14}:
\begin{align}
&\qquad L_{\epsilon,m,k}\left[n\right]=20{\log}_{10}\left(\frac{4\pi f_{c}d_{m,k}\left[n\right]}{c}\right)+\eta_\epsilon+S , \nonumber\\
&\qquad \forall m\in\{1,\ldots,M\}, \forall k\in\{1,\ldots,K\}, n\in\left\{0,\ldots,N\right\}, \nonumber\\
&\qquad \forall\epsilon\in \{LOS,NLOS\} \,, \label{eq7}
\end{align}where $f_{c}$ is the carrier frequency (Hz), $c$ is the speed of light (m/s), $\eta_\epsilon$ is an LOS and NLOS environment-related parameter, and $S$ represents the shadowing effect. Furthermore, the term $d_{m,k}[n]$ represents the distance between the $m$th UAV and the $k$th WN at the time instant $n$, given as
\begin{align}
&d_{m,k}\left[n\right]= \sqrt{\left\|\textbf q_m[n]-\textbf g_k\right\|_2^2+H^2} \,.\label{eq8}
\end{align}
According to \cite{Al-Hourani14}, the probability of occurring the LOS channel between the $m$th UAV and the $k$th WN at the time instant $n$ can be modelled as
\begin{align}
&\rho_{LOS,m,k}\left[n\right]\nonumber\\
&=\frac{1}{1+A \exp(-B(\frac{180}{\pi}\tan^{-1}(\frac{H}{\left\|\textbf q_m[n]-\textbf g_k\right\|_2})-A))} , \nonumber \\
& \forall m\in\{1,\ldots,M\}, \forall k\in\{1,\ldots,K\}, \forall n\in\left\{0,\ldots,N\right\} ,\label{eq9}
\end{align}
where the coefficients $A$ and $B$ depend on the operating environments \cite{J. Holis08}, and the probability of occurring the NLOS channel can be computed as $\rho_{NLOS,m,k}\left[n\right]=1-\rho_{LOS,m,k}\left[n\right]$. By combining the effect of large-scale and small-scale channel fading, we can obtain the channel gain between the $m$th UAV and the $k$th WN at the time instant $n$ as follows:
\begin{align}
H_{m,k}\left[n\right]={\check{L}}_{\epsilon,m,k}{\left|\chi_{m,k}[n]\right|}^2 , \label{eq10}
\end{align}
where ${\check{L}}_{\epsilon,m,k}\left[n\right]$ represents the linear scale of $L_{\epsilon,m,k}\left[n\right]$, $\chi_{m,k}[n]\sim CN(0,1) $ is a zero-mean complex Gaussian random variable with unit variance to describe the Rayleigh fading.

Next, we define the communication association variable between the $m$th UAV and the $k$th WN at the time instant $n$ as $a_{m,k}[n]$. The communication association variables are subject to the following three constraints:
\begin{align}
&a_{m,k}\left[n\right]\in\left\{0,1\right\}, \forall m\in\left\{1,\ldots,M\right\}, \forall k\in\left\{1,\ldots,K\right\}, \nonumber\\
& \qquad\qquad\qquad\qquad \forall n\in\left\{0,\ldots,N-1\right\} ; \label{eq11} \\
&\sum_{k=1}^{K}{a_{m,k}\left[n\right]\le1, \forall m\in\left\{1,\ldots,M\right\} ,\forall n\in\left\{0,\ldots,N-1\right\}}; \label{eq12}\\
&\sum_{m=1}^{M}{a_{m,k}\left[n\right]\le1, \forall k\in\left\{1,\ldots,K\right\},\forall n\in\left\{0,\ldots,N-1\right\}} , \label{eq13}
\end{align}where the constraint (\ref{eq11}) means that the communication association variables take binary values. If $a_{m,k}\left[n\right]=1$, the $k$th WN is associated with the $m$th UAV at the time instant $n$. The constraint (\ref{eq12}) confines that each UAV can serve at most one WN at each time instant, while each active WN can be only served by one UAV under the constraint (\ref{eq13}). Besides, if the $m$th UAV associates with one WN at the time instant $n$ for data transmissions, its position remains unchanged during the time interval $\delta_D$, leading to the following constraint:
\begin{align}
&a_{m,k}[n]\cdot\left\|\textbf q_m[n+1]-\textbf q_m[n]\right\|_2=0 ,  \nonumber \\
& \qquad\qquad \forall m\in\left\{1,\ldots,M\right\}, \forall n\in\left\{0,\ldots,N-1\right\}. \label{eq14}
\end{align}

Since the WNs utilize the same frequency band for uplink communications, each UAV suffers from the uplink multiuser interference problem. From (\ref{eq10}), the signal-to-interference plus noise power ratio (SINR) of the $k$th WN at the $m$th UAV and the $n$th time instant can be expressed as
\begin{align}
&\Gamma_{m,k}\left[n\right]=\frac{P_k\left[n\right]H_{m,k}\left[n\right]}{\sum_{i=1,i\neq k}^{K}{P_i\left[n\right]H_{m,i}\left[n\right]}+\sigma_n^2} ,\nonumber\\
& \forall m\in\left\{1,\ldots,M\right\}, \forall k\in\left\{1,\ldots,K\right\},\forall n\in\left\{0,\ldots,N-1\right\} , \label{eq15}
\end{align}
where $\sigma_n^2$ is the power of additive white Gaussian noise, and $P_k\left[n\right]$ is the uplink transmit power of the $k$th WN at the time slot $n$. Accordingly, the achievable sum rate of the $k$th WN over the whole UAV task period can be calculated as
\begin{align}
&R_k =\sum_{n=0}^{N-1} \underbrace{\sum_{m=1}^{M}{a_{m,k}\left[n\right]W
{\log}_2\left(1+\Gamma_{m,k}\left[n\right]\right)}}_{\triangleq R_{k,n}}, \nonumber\\
& \qquad\qquad\qquad\qquad \forall k\in\left\{1,\ldots,K\right\},\label{eq16}
\end{align}where $W$ is the system bandwidth, and $R_{k,n}$ is the data rate of the $k$th WN at the time instant $n$. 

Since each WN relies on the harvested solar energy in a finite-capacity battery for uplink transmission, the transmit power of a WN is constrained by its harvested energy and battery capacity. Let $E_k\left[n\right]$ and $B_k\left[n\right]$ represent the harvested energy and battery level of the $k$th WN at the time instant $n$. {For simplicity, we assume that the battery level of WN is zero at the time instant $n = 0$.} Therefore, the residual energy in the battery of the $k$th WN at the time instant $n$ can be described as $B_k\left[n\right]=\sum_{l=0}^{n}{E_k\left[l\right]}-\delta_D\sum_{l=0}^{n}{P_k\left[l\right]}$, and the battery power evolution from the time $n$ to $\left(n+1\right)$ can be described as $B_k\left[n+1\right]=B_k\left[n\right]+E_k\left[n+1\right]$. Since the battery power is limited by $0\le B_k\left[n\right]\le B_{max\ }$, it then yields two constraints about the uplink transmit power, harvested solar energy, and battery power, namely energy causality and battery storage constraints:
\begin{align}
&\qquad\delta_D\sum_{l=0}^n P_k[l]\leq \sum_{l=0}^nE_k[l] ,\forall k\in\left\{1,\ldots,K\right\},\nonumber\\
& \qquad\qquad\qquad\qquad \forall n\in\{0,\ldots,N-1\} ; \label{eq17}\\
&\sum_{l=0}^{n+1}{E_k\left[l\right]}-\delta_D\sum_{l=0}^{n}{P_k\left[l\right]}\le B_{\max}, \forall k\in\left\{1,\ldots,K\right\}, \nonumber \\
& \qquad\qquad\qquad\qquad \forall n\in\left\{0,\ldots,N-1\right\}, \label{eq18}
\end{align}where $B_{\max}$ is the maximum battery storage capacity of WNs. Here, the energy causality constraint (\ref{eq17}) mandates that the harvested energy at each WN cannot be used until it arrives. The battery storage constraint (\ref{eq18}) indicates that for each WN, the total amount of harvested energy minus the energy expenditure for data transmission cannot exceed the maximum battery storage capacity at each time instant. 

\subsection{Problem Formulation}
In order to achieve a fair communication service for multiple WNs during the entire UAV task period, the design goal of this paper is to maximize the worst sum rate among $K$ WNs. The joint design problem of the trajectories of UAVs, communication association between UAVs and WNs, and transmit power of WNs with solar EH can be formulated as follows:
\begin{align}
\textbf{(P1)}\;&\max_{\{\textbf q_m\left[n\right],P_k\left[n\right],a_{m,k}\left[n\right],\forall m,k,n\}} \min_{k=1,\ldots,K}{R_k}  \nonumber\\
s.t.\;& (\ref{eq4}), (\ref{eq5}), (\ref{eq6}), (\ref{eq11}), (\ref{eq12}), (\ref{eq13}), (\ref{eq14}), (\ref{eq17}), (\ref{eq18}) \nonumber.
\end{align}
It is worth noting that if the $k$th WN is not associated with any UAV, the optimal solution to the uplink transmit power $P_k[n]$ of the $k$th WN must be equal to zero; otherwise, it causes the interference problem and degrades the overall sum rate performance of the network.

In this optimization problem, the UAV needs to know the CSI $H_{m,k} [n]$ and ESI $E_k [n]$ not only in the current time instant but also in the future time instants. However, in real applications, it is impractical to acquire the full (past and future) information for carrying out the optimization problem. In response to this design challenge, a RL method can be utilized to perform dynamic optimization for the UAV flight direction, UAV and WN association, and uplink transmit power, solely based on the current battery state information (BSI) of WNs and the current CSI between UAVs and WNs.

{As compared with the convex optimization approach, the conventional RL is exempt from the future instantaneous ESIs and CSIs but only depends on the current state information. However, it requires the UAV to repeatedly take various actions in each state of the environment, which results in a long exploration learning time and a slow update of Q-values. The learning time and convergence performance become even worse for multi-UAV communication networks with multi-EH WNs, since the number of system states and actions increases exponentially and many inefficient actions may be executed during the learning process. For this reason, in this paper, we will first investigate an offline convex optimization method to find the best offline flight trajectory for the UAVs by only applying the ``average" LOS/NLOS channel gain and EH profile. The obtained solution is then served as the reference solution for the online RL. The main idea is to mark out a flight corridor based on the reference solution to guide the UAV flight actions in the online learning for reducing the learning time of the conventional RL while improving the system performance of the offline convex optimization approach.}

\section{Offline Convex Optimization Design of Multi-UAV Systems with Multiple EH WNs}\label{offline}
In this section, we develop an offline convex optimization approach by using the average LOS/NLOS channel gain value on the UAV trajectory $\bar{H}_{m,k}[n]=\mathbb{E}[H_{m,k} [n]]$ and the average value of the past solar EH time series $\bar{E}_k[n]=\mathbb{E}[E_k[n]]$ to replace the instantaneous information $H_{m,k}[n]$ and $E_k[n]$ in the optimization problem (\textbf{P1}), respectively. It is worth noting that although the statistical average of solar EH profiles is difficult to obtain, it can be acquired by numerically averaging real solar energy historical data. Besides, from (\ref{eq7})--(\ref{eq10}), the average LOS/NLOS channel gain $\bar{H}_{m,i}[n]$ can be derived in terms of the UAV position $\textbf q_m[n]$ as follows: 
\begin{align}
    &\bar H_{m,i} [n]\nonumber\\
    &=\rho_{LOS,m,i} [n] \check L_{LOS,m,i} [n]+\rho_{NLOS,m,i}[n]\check L_{NLOS,m,i}[n]\nonumber\\
    &=\left(\frac{c}{4\pi f_c d_{m,i}[n]}\right)^2\times 10^{\frac{-S}{10}}\nonumber\\
    &\;\;\times\left(\rho_{LOS,m,i}[n]\times10^{\frac{-\eta_{LOS}}{10}}+\rho_{NLOS,m,i}[n]\times10^{\frac{-\eta_{NLOS}}{10}}\right)\nonumber\\
    &=C_1\times C_2 \times d^{-2}_{m,i}[n]\times\rho_{LOS,m,i}[n]+C_3\times  d^{-2}_{m,i}[n]\nonumber\\
    &=C_1\times C_2 \times \frac{1}{\left\|\textbf q_m[n]-\textbf g_i\right\|_2^2+H^2}\nonumber\\
    &\;\;\times\frac{1}{1+A exp\left(-B\left(\frac{180}{\pi}\tan^{-1}(\frac{H}{\left\|\textbf q_m[n]-\textbf g_i\right\|_2})-A\right)\right)}\nonumber\\
    &\;\;+C_3\times \frac{1}{\left\|\textbf q_m[n]-\textbf g_i\right\|_2^2+H^2}\;,\label{average_H}
\end{align}
where $C_1$, $C_2$, and $C_3$ are constants and all greater than zero:
\begin{align}
    &C_1=\left(\frac{c}{4\pi f_c}\right)^2\times10^{\frac{-S}{10}}>0  ;\label{eq33}\\
    &C_2=\left(10^{\frac{-\eta_{LOS}}{10}}-10^{\frac{-\eta_{NLOS}}{10}}\right)>0 ;\label{eq34}\\
    &C_3=\left(\frac{c}{4\pi f_c}\right)^2\times10^{\frac{-S}{10}}\times10^{\frac{-\eta_{NLOS}}{10}}>0 .\label{eq39_c3}
\end{align}
Therefore, the objective function considered in the offline optimization problem is given by
\begin{align}
&\bar R_k=\sum_{n=0}^{N-1}\sum_{m=1}^{M}{a_{m,k}\left[n\right]W
{\log}_2(1+\bar\Gamma_{m,k}[n])}, \nonumber\\
&\;\;\;\;\;\;\;\;\;\;\;\;\;\;\;\;\;\; \forall k \in \left\{ 1,\ldots,K \right\} ,\label{eq26}
\end{align}where $\bar\Gamma_{m,k}[n]$ is the SINR calculated throughout the average channel gain:
\begin{align}
&\bar\Gamma_{m,k}\left[n\right]=\frac{P_k\left[n\right]\bar H_{m,k}\left[n\right]}{\sum_{i=1,i\neq k}^{K}{P_i\left[n\right]\bar H_{m,i}\left[n\right]}+\sigma_n^2},\forall m\in\left\{1,\ldots,M\right\}\nonumber\\
&  \;\;\;\;\; \;\;\;\;\;\;,\forall k\in\left\{1,\ldots,K\right\},\forall n\in\left\{0,\ldots,N-1\right\} .\label{eq27}
\end{align}
From (\textbf{P1}), the offline joint design problem is formulated as
\begin{align}
\textbf{(P2)}\;&\max_{\{\textbf q_m\left[n\right],P_k\left[n\right],a_{m,k}\left[n\right],\forall m,k,n\}} \min_{k=1,\ldots,K}{\bar R_k}  \nonumber\\
s.t.\;& (\ref{eq4}), (\ref{eq5}), (\ref{eq6}), (\ref{eq11}), (\ref{eq12}), (\ref{eq13}), (\ref{eq14}), (\ref{eq17}), (\ref{eq18}) \nonumber.
\end{align}
Since the objective function of the joint design problem (\textbf{P2}) and the constraints (\ref{eq6}), (\ref{eq11}) and (\ref{eq14}) are non-convex for the three design variables $\textbf q_m[n]$, $a_{m,k}[n]$ and $P_k[n]$. To overcome this issue, we adopt the alternative optimization to decompose the joint optimization problem into three subproblems to optimize the UAV flight trajectory $\{\textbf q_{m} [n],\forall m,n\}$, communication association $\{a_{m,k} [n],\forall m,k,n\}$, or power control $\{P_k [n],\forall k,n\}$ under the fixed values of other variables. Nevertheless, the three subproblems are still non-convex: (i) the communication association subproblem is a non-convex mixed integer programming, and we will relax the integer constraint (\ref{eq11}), i.e., $\{a_{m,k} [n]\in\{0,1\},\forall m,k,n\}$ to $\{0\leq a_{m,k} [n]\leq1,\forall m,k,n\}$ in order to convert the problem into a linear programming problem; (ii) the UAV flight trajectory subproblem and the power control subproblem are also non-convex due to the objective function and the constraint (\ref{eq6}), and we will use SCA methods to transform these two non-convex subproblems into convex ones.

\subsection{UAV-WN Communication Association Subproblem}

Given the transmit power control $P_k [n]$ of the WNs and the UAV flight trajectory ${\textbf q}_m [n]$ in (\textbf{P2}), the optimization can be performed for the communication association. By relaxing the integer constraint (\ref{eq11}) and introducing an auxiliary variable $\zeta_a$, the subproblem becomes an equivalent epigraph form:
\begin{align}
\textbf{(P3)}\;&\qquad\qquad\qquad\max_{\{\zeta_a,a_{m,k}\left[n\right],\forall m,k,n\}}{\zeta_a}\nonumber\\
s.t.\;         &\sum_{n=0}^{N-1}\sum_{m=1}^{M}{a_{m,k}\left[n\right]W
{\log}_2(1+\bar\Gamma_{m,k}[n])}\geq\zeta_a,\nonumber\\
& \;\;\;\;\;\;\;\;\;\;\;\;\;\;\;\;\;\;\;\;\;\;\;\;\;\;\;\;\;\;\;\;\;\;\;\;\forall k\in\{1,...,K\},\label{eq28} \\
               &0\leq a_{m,k}[n]\leq 1,\forall m,\forall n\in\left\{0,\ldots,\ N-1\right\},\label{eq29}\\
               &(\ref{eq12}),(\ref{eq13})\nonumber.
\end{align}
The subproblem (\textbf{P3}) now becomes linear programming which can be directly optimized by using off-the-shelf optimization software (CVX) \cite{M. Grant16}. Since the obtained solution $a_{m,k}^\star \left[n\right]$ in (\textbf{P3}) may take a value between $0$ and $1$, it is quantized to $1$, if $a_{m,k}^\star \left[n\right]\geq 0.5$; otherwise, we set $a_{m,k}^\star \left[n\right]$ to $0$.

\subsection{UAV Flight Trajectory Subproblem}
Given the communication association $a_{m,k}[n]$ and the power control $P_k[n]$ of WNs, the UAV flight trajectory subproblem can be obtained from the optimization problem (\textbf{P2}) and represented as an epigraph form with the introduction of an auxiliary variable $\zeta_q$:
\begin{align}
  \textbf{(P4)}\;&\max_{\{\zeta_q,\textbf{q}_m\left[n\right],\forall m,k,n\}}{\zeta_q}\nonumber\\
  s.t.\;
  &\bar R_k\geq\zeta_q, \forall k\in\{1,...,K\}, \label{eq30}\\
  &(\ref{eq4}),(\ref{eq5}),(\ref{eq6}),(\ref{eq14})\nonumber.
\end{align}
Observing the constraint (\ref{eq30}), we know that the SINR $\bar\Gamma_{m,k}[n]$ in (\ref{eq27}) contains $\bar H_{m,k}[n]$ which is related to the multiple UAV flight trajectory variables $\textbf{q}_m [n]$. This makes the constraint (\ref{eq30}) non-convex. In addition, the minimum safe distance constraint (\ref{eq6}) is also non-convex. As a result, the subproblem (\textbf{P4}) is a non-convex optimization problem, which cannot be directly solved by convex optimization tools. In the following, we resort to the SCA method to convexify the non-convex constraints considering LOS/NLOS channels for multiple UAVs. Note that the SCA method was applied in {\cite{C. You20}\cite{Q. Wu18}}, whereas these two works merely considered the design in LOS channels or single-UAV environments and cannot be directly applied to our work.
\begin{figure*}[t]
\begin{align}
\bar R_k&=\sum_{n=0}^{N-1}\sum_{m=1}^{M}{a_{m,k}\left[n\right]W\left(\underbrace{{\log}_2\left( \sum_{i=1}^{K}P_i[n]\bar H_{m,i}[n]+\sigma_n^2 \right)}_{\triangleq \check R_{1m}[n]}  \underbrace{-{\log}_2\left(\sum_{i=1,i\neq k}^{K}P_i[n]\bar H_{m,i}[n]+\sigma_n^2\right)}_{\triangleq \check R_{2m,k}[n]}\right)}\label{eq31}
\end{align}
\hrule
\end{figure*}

{First, we expand the left-hand side of the constraint (\ref{eq30}) into the difference of two logarithmic functions, as shown in (\ref{eq31}) at the top of the next page, where we define $\check R_{1m}[n] \triangleq {\log}_2\left(\sum_{i=1}^{K}P_i[n]\bar H_{m,i}[n]+\sigma_n^2\right)$ and $\check R_{2m,k}[n] \triangleq -{\log}_2\left(\sum_{i=1,i\neq k}^{K}P_i[n]\bar H_{m,i}[n]+\sigma_n^2\right)$. Our goal is to find a lower bound concave function for $\bar R_k$ in (\ref{eq31}) in order to transform the constraint (\ref{eq30}) into solvable convex constraints. Below we elaborate on the ways to find concave lower bounds for $\check R_{1m}[n]$ and $\check R_{2m,k}[n]$ in terms of $\textbf{q}_m [n]$, followed by the transformation of the subproblem (\textbf{P4}) into a solvable convex optimization problem through these derived lower bounds. 

\subsubsection{A concave lower bound for $\check R_{1m}[n]$}
We define two variables $X_{m,i}$ and $Y_{m,i}$, given by
\begin{align}
    &X_{m,i}[n]= \left\|\textbf q_m[n]-\textbf g_i\right\|_2^2+H^2, \forall m\in\{1,...,M\},\nonumber\\
    &\;\;\;\;\;\;\;\forall i\in\{1,...,K\},\forall n\in\{0,...,N-1\}.\label{eq35add}\\
    &Y_{m,i}[n]= 1+A e^{-B\left(\frac{180}{\pi}\tan^{-1}\left(\frac{H}{\left\|\textbf q_m[n]-\textbf g_i\right\|_2}\right)-A\right)},\nonumber\\
    &\;\;\;\forall m\in\{1,...,M\},\forall i\in\{1,...,K\},\forall n\in\{0,...,N-1\}.\label{eq36add}\
\end{align}
The following lemma is provided.
\begin{lemma}\label{log_lower_bound__convex}
The function $\check R_{1m}[n]={\log}_2\big(\sum_{i=1}^{K}P_i[n]$ $\bar H_{m,i}[n]+\sigma_n^2\big)$ is convex in $X_{m,i}[n]$ and $Y_{m,i}[n]$ for all $i$.
\end{lemma}
\begin{proof}
See Appendix \ref{Lemma_1} for the detailed proof.
\end{proof}
Since our purpose is to find a lower bound concave function for $\check R_{1m}[n]$ in (\ref{eq31}), we use the property that the first-order Taylor expansion of a convex function at any point is always a lower bound of that convex function. By using Lemma \ref{log_lower_bound__convex}, a theorem for the first-order Taylor expansion of $\check R_{1m}[n]$ is then provided as follows.
\begin{theorem}\label{first_order_R_lowerbound}
Given any ${\textbf q}_m[n]={\textbf q}^r_m[n]$, the first-order Taylor expansion of $\check R_{1m}[n]$ can be derived and served as a lower bound:
\begin{align}
    &\check R_{1m}[n] \triangleq {\log}_2\left(\sum_{i=1}^{K}P_i[n]\bar H_{m,i}[n]+\sigma_n^2\right) \nonumber\\
    &\geq {log}_2 \left(\sum_{i=1}^K P_i[n]\left(\frac{C_1 C_2}{X_{m,i}^r[n]Y_{m,i}^r[n]}+\frac{C_3}{X_{m,i}^r[n]}\right)+\sigma_n^2\right)\nonumber\\
    &\;\;\;+\sum_{i=1}^K O_{m,i}[n]\left( X_{m,i}[n]-X_{m,i}^r[n]\right)\nonumber\\
    &\;\;\;+\sum_{i=1}^K G_{m,i}[n]\left( Y_{m,i}[n]-Y_{m,i}^r[n]\right)\nonumber\\
    &\triangleq\check R_{1m}^{lb}[n],\forall m\in\{1,...,M\},\forall n\in\{0,...,N-1\},\label{eq41}
\end{align}
where we define
\begin{align}
    &X_{m,i}^r[n]=\|\textbf q_m^r[n]-\textbf g_i\|^2+H^2; \label{eq41_1}\\
    &Y_{m,i}^r[n]=1+A e^{-B\left(\frac{180}{\pi}\tan^{-1}\left(\frac{H}{\left\|\textbf q^r_m[n]-\textbf g_i\right\|_2}\right)-A\right)}; \label{eq41_2}\\
    &O_{m,i}[n]=\frac{-P_i[n]\left(\frac{C_1 C_2+C_3 Y_{m,i}^r[n]}{\left(X_{m,i}^r[n]\right)^2 Y_{m,i}^r[n]}\right)}{\left(\sum_{j=1}^K P_j[n]\left(\frac{\frac{C_1 C_2}{Y_{m,i}^r[n]}+C_3}{X_{m,i}^r[n]}\right)+\sigma_n^2\right)ln(2)}; \label{eq41_3}\\
    &G_{m,i}[n]=\frac{-P_i[n]\left(\frac{C_1 C_2}{X_{m,i}^r[n] \left(Y_{m,i}^r[n]\right)^2}\right)}{\left(\sum_{j=1}^K P_j[n]\left(\frac{\frac{C_1 C_2}{Y_{m,i}^r[n]}+C_3}{X_{m,i}^r[n]}\right)+\sigma_n^2\right)ln(2)}. \label{eq41_4}
\end{align}
\end{theorem}
\begin{proof}
See Appendix \ref{first_order_R_lowerbound_proof} for the detailed proof.
\end{proof}
It can be observed from Theorem \ref{first_order_R_lowerbound} that the term $O_{m,i}[n]\left( X_{m,i}[n]-X_{m,i}^r[n]\right)$ is concave in $\textbf q_m[n]$, since $O_{m,i}[n]\leq 0$ and $X_{m,i}[n]= \left\|\textbf q_m[n]-\textbf g_i\right\|_2^2+H^2$ is a square norm function of $\textbf q_m[n]$. However, the term $G_{m,i}[n]\left( Y_{m,i}[n]-Y_{m,i}^r[n]\right)$ is neither concave nor convex, and we further provide the following theorem to transform this term into a concave function.
\begin{theorem}\label{theorem_Y_upperbound}
Given any ${\textbf q}_m[n]={\textbf q}^r_m[n]$, $Y_{m,i}[n]$ can be upper bounded by
\begin{align}
    &Y_{m,i}[n]\leq 1+Aexp\Biggl( -B\Biggl(\frac{180}{\pi}\biggl(\tan^{-1}\Bigl(\frac{H}{\sqrt{U_{m,i}^r[n]}}\Bigr)\nonumber\\
    &\qquad\quad\;\;\;-\frac{H\bigl(\left\|\textbf q_m[n]-\textbf g_i\right\|_2^2-U_{m,i}^r[n]\bigr)}{2\sqrt{U_{m,i}^r[n]}\left(H^2+U_{m,i}^r[n]\right)}\biggr)-A\Biggr)\Biggr)\nonumber\\
    &\qquad\quad\triangleq Y_{m,i}^{up}[n],\forall m\in\{1,...,M\},\forall i\in\{1,...,K\},\nonumber\\
    &\qquad\qquad\qquad\qquad\forall n\in\{0,...,N-1\},\label{Y_upperbound}
\end{align}where $U^r_{m,i}[n]=\left\|\textbf q^r_m[n]-\textbf g_i\right\|_2^2$ and $Y_{m,i}^{up}[n]$ is a convex function in terms of $\textbf q_m[n]$.
\end{theorem}
\begin{proof}
See Appendix \ref{upper_bound_Y_proof} for the detailed proof.
\end{proof}

By applying Theorem \ref{theorem_Y_upperbound}, we replace $Y_{m,i}[n]$ in $\check R^{lb}_{1m}[n]$ by $Y^{up}_{m,i}[n]$. Since $G_{m,i}[n]\leq0$, the term $G_{m,i}[n]\left( Y^{up}_{m,i}[n]-Y_{m,i}^r[n]\right)$ is a concave lower bound in $\textbf q_m[n]$. Hence, the concave lower bound of $\check R^{lb}_{1m}[n]$ in (\ref{eq41}) can be derived as follows:
\begin{align}
    &\check R^{lb}_{1m}[n]\nonumber\\
    &\geq {log}_2 \left(\sum_{i=1}^K P_i[n]\left(\frac{C_1 C_2}{X_{m,i}^r[n]Y_{m,i}^r[n]}+\frac{C_3}{X_{m,i}^r[n]}\right)+\sigma_n^2\right)\nonumber\\
    &\;\;\;+\sum_{i=1}^K O_{m,i}[n]\left( X_{m,i}[n]-X_{m,i}^r[n]\right)\nonumber\\
    &\;\;\;+\sum_{i=1}^K G_{m,i}[n]\left( Y^{up}_{m,i}[n]-Y_{m,i}^r[n]\right)\nonumber\\
    &\triangleq\check R_{1m}^{llb}[n],\forall m\in\{1,...,M\},\forall n\in\{0,...,N-1\}.\label{R1m_llb}
\end{align}

\subsubsection{A concave lower bound for $\check R_{2m,k}[n]$}
In order to find a concave lower bound for $\check R_{2m,k}[n]$, we introduce two auxiliary variables $\tilde X_{m,i}[n]$ and $\tilde Y_{m,i}[n]$, which satisfy the following constraints:
\begin{align}
    &\tilde X_{m,i}[n]\leq \left\|\textbf q_m[n]-\textbf g_i\right\|_2^2+H^2, \forall m\in\{1,...,M\},\nonumber\\
    &\;\;\;\forall i\in\{1,...,K\},\forall n\in\{0,...,N-1\}.\label{eq39add}\\
    &\tilde Y_{m,i}[n]\leq 1+A e^{-B\left(\frac{180}{\pi}\tan^{-1}\left(\frac{H}{\left\|\textbf q_m[n]-\textbf g_i\right\|_2}\right)-A\right)},\nonumber\\
    &\;\;\;\forall m\in\{1,...,M\},\forall i\in\{1,...,K\},\forall n\in\{0,...,N-1\}.\label{eq40add}\
\end{align}
Thus, an upper bound $\bar H_{m,i}^{up}[n]$ for $\bar H_{m,i}[n]$ in (\ref{eq31}) can be derived as 
\begin{align}
    \bar H_{m,i}[n]&\leq C_1 \times C_2 \times \frac{1}{\tilde X_{m,i}[n]}\times \frac{1}{\tilde Y_{m,i}[n]}+ C_3\times \frac{1}{\tilde X_{m,i}[n]}\nonumber\\
                        &\triangleq \bar H_{m,i}^{up}[n],\forall m\in\{1,...,M\},\forall i\in\{1,...,K\},\nonumber\\
                        &\qquad\qquad\qquad\qquad\forall n\in\{0,...,N-1\}, \label{eq39}
\end{align}where $C_1$, $C_2$ and $C_3$ are constants defined in (\ref{eq33}), (\ref{eq34}), and (\ref{eq39_c3}), respectively.

\begin{lemma}
The upper bound $\bar H_{m,i}^{up} [n]$ is a convex function for the two auxiliary variables $\tilde X_{m,i}[n]$ and $\tilde Y_{m,i}[n]$, $\forall$ $m\in\{1,...,M\}$, $\forall i\in\{1,...,K\}$, $\forall n\in\{0,...,N-1\}$.\label{lemma_Hup_convex}
\end{lemma}
\begin{proof}
According to the proof in Lemma \ref{log_lower_bound__convex}, we know that the function $\ln(\frac{C_1C_2}{xy}+\frac{C_3}{x})$ is convex in $(x,y)$ if $C_1, C_2, C_3>0$. By using the fact that $e^{g(x)}$ is convex if $g(x)$ is convex {\cite{S. Boyd04}}, it can be shown that $\frac{C_1C_2}{xy}+\frac{C_3}{x}= e^{\ln(\frac{C_1C_2}{xy}+\frac{C_3}{x})}$ is also convex. Hence, the proof is completed.
\end{proof}

\begin{theorem}\label{theorem_R_2m_lowerbound}
With (\ref{eq39add}) and (\ref{eq40add}), the function $\check R_{2m,k}[n] = -{\log}_2\big( \sum_{i=1,i\neq k}^{K}P_i[n]\bar H_{m,i}[n]+\sigma_n^2\big)$ can be lower bounded by  
\begin{align}
&\check R_{2m,k}[n] \geq -{\log}_2\left(\sum_{i=1,i\neq k}^{K}P_i[n]\bar H^{up}_{m,i}[n]+\sigma_n^2\right) \nonumber \\
& \triangleq \check R_{2m,k}^{lb}[n], \forall m\in\{1,...,M\}, \forall n\in\{0,...,N-1\} \,,
\end{align}where $\check R_{2m,k}^{lb}[n]$ is a concave function in terms of $\tilde X_{m,i}[n]$ and $\tilde Y_{m,i}[n]$.
\end{theorem}

\begin{proof}
With (\ref{eq39add}) and (\ref{eq40add}), we have $\bar H_{m,i}[n] \leq \bar H_{m,i}^{up} [n]$ in (\ref{eq39}). By replacing $\bar H_{m,i}[n]$ in $\check R_{2m,k}[n]$ with the upper bound $\bar H_{m,i}^{up} [n]$, it then yields $\check R_{2m,k}[n] \geq$ $-{\log}_2\left(\sum_{i=1,i\neq k}^{K}P_i[n]\bar H^{up}_{m,i}[n]+\sigma_n^2\right)$. Moreover, it can be shown that $\check R_{2m,k}^{lb}[n]$ is a concave function in terms of $\tilde X_{m,i}[n]$ and $\tilde Y_{m,i}[n]$ by applying the similar proof in Lemma \ref{log_lower_bound__convex}.
\end{proof}

\subsubsection{Transformation of subproblem (\textbf{P4}) into a convex problem}

By adopting the derived lower bounds in Theorem \ref{first_order_R_lowerbound}, Theorem \ref{theorem_R_2m_lowerbound} and (\ref{R1m_llb}), a concave lower bound for  $\bar R_k$ can be computed as
\begin{align}
  \bar R_k \geq  \sum_{n=0}^{N-1}\sum_{m=1}^{M}a_{m,k}\left[n\right]W  \left(\check R_{1m}^{llb}[n]+ \check R_{2m,k}^{lb}[n] \right)\label{eq42} .
\end{align}
We then transform the subproblem (\textbf{P4}) into a convex one by replacing $\bar R_k$ in the constraint (\ref{eq30}) with its lower bound (\ref{eq42}) and inserting the two imposed constraints (\ref{eq39add}) and (\ref{eq40add}):
\begin{align}
\textbf{(P5)}\;&\qquad\qquad\max_{\{\zeta_q,\textbf{q}_m\left[n\right],\tilde X_{m,i}[n],\tilde Y_{m,i}[n],\forall m,i,n\}}{\zeta_q}\nonumber\\
  s.t.\;
  &\sum_{n=0}^{N-1}\sum_{m=1}^{M}a_{m,k}\left[n\right]W  \left(\check R_{1m}^{llb}[n] + \check R_{2m,k}^{lb}[n] \right)\geq\zeta_q,\nonumber\\
  &\;\qquad\qquad\qquad\forall k\in\{1,...,K\},\label{eq43}\\
  &(\ref{eq4}),(\ref{eq5}),(\ref{eq6}),(\ref{eq14}),(\ref{eq39add}),(\ref{eq40add})\nonumber.
\end{align}
In (\textbf{P5}), the constraint (\ref{eq43}) now becomes convex, since the left-hand side of (\ref{eq43}) is a concave function in terms of the variables $\textbf{q}_m\left[n\right]$, $\tilde X_{m,i}[n]$ and $\tilde Y_{m,i}[n]$ according to (\ref{eq42}). The constraints (\ref{eq39add}) and (\ref{eq40add}), which are due to the introduction of the auxiliary variables $\tilde X_{m,i}[n]$ and $\tilde Y_{m,i }[n]$ in (\ref{eq39}), are also included in (\textbf{P5}) for guaranteeing the lower bound relationship of (\ref{eq42}) and thus the legitimacy of the constraint (\ref{eq43}). Due to the lower bound relationship in (\ref{eq42}), one can also conclude that the feasible set of the subproblem (\textbf{P5}) is a subset of the feasible set of the original UAV flight trajectory subproblem (\textbf{P4}).
}

However, the subproblem (\textbf{P5}) is still not a convex optimization problem, since the constraints (\ref{eq6}), (\ref{eq39add}) and (\ref{eq40add}) are non-convex. We in turn apply the first-order Taylor expansion to deal with these constraints and transform them into convex ones. To make the constraint (\ref{eq6}) easier to handle, we first square the both sides of the constraint (\ref{eq6}):
\begin{align}
    &\left\|\textbf q_m[n]-\textbf q_j[n]\right\|_2^2\geq D^2_{min} ,\nonumber\\
    &\forall n \in\{ 1,\ldots,N-1\},\forall m,j\in\{1,\ldots,M\},m\neq j.\label{eq44}
\end{align}
Since $\left\|\textbf q_m[n]-\textbf q_j[n]\right\|_2^2$ is a convex function with respect to $\textbf q_m[n]$ and $\textbf q_j[n]$, the first-order Taylor expansion of $\left\|\textbf q_m[n]-\textbf q_j[n]\right\|_2^2$ at given points $\textbf q_m^r[n]$ and $\textbf q_j^r[n]$ can be derived and served as a lower bound:
\begin{align}
&    \left\|\textbf q_m[n]-\textbf q_j[n]\right\|_2^2\geq\left\|\textbf q^r_m[n]-\textbf q^r_j[n]\right\|_2^2 \nonumber\\
                      &\;\;\;\;\;\;+2(\textbf q^r_m[n]-\textbf q^r_j[n])^T(\textbf q_m[n]-\textbf q^r_m[n]) \nonumber \\
                      &\;\;\;\;\;\;-2(\textbf q^r_m[n]-\textbf q^r_j[n])^T(\textbf q_j[n]-\textbf q^r_j[n]). \label{eq_square_q_lowerbound}
\end{align}
From (\ref{eq_square_q_lowerbound}), the constraint (\ref{eq44}) is then replaced by the following new constraint:
\begin{align}
    &D^2_{min}\leq\left\|\textbf q^r_m[n]-\textbf q^r_j[n]\right\|_2^2\nonumber\\
    &\qquad\;\;\;\;+2(\textbf q^r_m[n]-\textbf q^r_j[n])^T(\textbf q_m[n]-\textbf q^r_m[n])\nonumber\\
    &\qquad\;\;\;\;-2(\textbf q^r_m[n]-\textbf q^r_j[n])^T(\textbf q_j[n]-\textbf q^r_j[n]),\nonumber\\
    &\forall n \in\{ 1,\ldots,N-1\},\forall m,j\in\{1,\ldots,M\},m\neq j.\label{eq45}
\end{align}
Note that the solutions that satisfy the constraint (\ref{eq45}) must also satisfy the constraint (\ref{eq6}) due to the lower bound relationship in (\ref{eq_square_q_lowerbound}).

Similarly, the function $\left\|\textbf q_m[n]-\textbf g_i\right\|_2^2+H^2$ in the constraint (\ref{eq39add}) is a convex function with respect to the trajectory variable $\textbf q_m [n]$, and its lower bound is obtained by the first-order Taylor expansion at a given point $\textbf q_m^r [n]$, as follows:
\begin{align}
 &   \left\|\textbf q_m[n]-\textbf g_i\right\|_2^2+H^2\geq\left\|\textbf q_m^r[n]-\textbf g_i\right\|_2^2+H^2\nonumber\\
    &\;\;\;\;\;\;\;\;\;+2(\textbf q^r_m[n]-\textbf g_i)^T(\textbf q_m[n]-\textbf q^r_m[n]).
    \label{eq_distance_UAV_WN_lowerbound}
\end{align}
By using (\ref{eq_distance_UAV_WN_lowerbound}), the constraint (\ref{eq39add}) is then replaced by the following new constraint:
\begin{align}
    & \tilde X_{m,i}[n]\leq\left\|\textbf q_m^r[n]-\textbf g_i\right\|_2^2+H^2\nonumber\\
    &\;\;+2(\textbf q^r_m[n]-\textbf g_i)^T(\textbf q_m[n]-\textbf q^r_m[n]), \forall n \in\{ 1,\ldots,N-1\},\nonumber\\
    &\;\;\;\;\;\;\forall m\in\{1\ldots,M\},\forall i\in\{1,\ldots,K\}.\label{eq46}
\end{align}

The next step is to convexify the constraint (\ref{eq40add}) because the right-hand side of (\ref{eq40add}) is neither convex nor concave with respect to the trajectory variable $\textbf q_m[n]$. To achieve this, we first define an angle elevation variable $\theta_{m,i}[n]$ which satisfies the following imposed constraint:
\begin{align}
    \theta_{m,i}[n]\geq\frac{180}{\pi}\tan^{-1}\left(\frac{H}{\left\|\textbf q_m[n]-\textbf g_i\right\|_2}\right). \label{eq_imposed_angle_elevation}
\end{align}
This implies that
\begin{align}
    Ae^{-B\left(\theta_{m,i}[n]-A\right)}\leq A e^{-B\left(\frac{180}{\pi}\tan^{-1}\left(\frac{H}{\left\|\textbf q_m[n]-\textbf g_i\right\|_2}\right)-A\right)}, \label{eq_exp_angle_lowerbound}
\end{align}
where $A$ and $B$ are the non-negative coefficients defined in (\ref{eq9}). In addition, it is worth mentioning that the function $A \exp(-B(\theta_{m,i}[n]-A))$ is convex with respect to $\theta_{m,i}[n]$, and its first-order Taylor expansion at the point $\theta_{m,i}^r[n]=\frac{180}{\pi}\tan^{-1}\left(\frac{H}{\left\|\textbf q_m^r[n]-\textbf g_i\right\|_2}\right)$ can be derived as
\begin{align}
    &A e^{-B\left(\theta_{m,i}[n]-A\right)}\geq A e^{-B\left(\theta^r_{m,i}[n]-A\right)} \nonumber\\
    &\;\;\;\;\;\;\;\;\;+\left(-AB e^{B\left(A-\theta^r_{m,i}[n]\right)}\right) \left(\theta_{m,i}[n]-\theta^r_{m,i}[n]\right). \label{exp_angle_taylor_expansion}
\end{align}
Applying the first-order expression of (\ref{exp_angle_taylor_expansion}) into (\ref{eq40add}) yields a convex constraint:
\begin{align}
    &\tilde Y_{m,i}[n]\leq 1+A e^{-B\left(\theta^r_{m,i}[n]-A\right)}+\left(-A B e^{B\left(A-\theta^r_{m,i}[n]\right)}\right)\nonumber\\
    &\;\;\;\;\;\;\times\left(\theta_{m,i}[n]-\theta^r_{m,i}[n]\right), \forall n \in\{ 1,\ldots,N-1\},\nonumber\\
    &\;\;\;\;\;\;\;\;\;\forall m\in\{1\ldots,M\},\forall i\in\{1,\ldots,K\}.\label{eq48}
\end{align}
This new convex constraint (\ref{eq48}) ensures the satisfaction of the constraint (\ref{eq40add}) according to the lower bound relationship in (\ref{eq_imposed_angle_elevation}) and (\ref{exp_angle_taylor_expansion}). 

However, the angle elevation angle constraint (\ref{eq_imposed_angle_elevation}) is non-convex in terms of the flight trajectory variable $\textbf q_m[n]$. In the following, we attempt to convexify this non-convex constraint. According to the fact that $\tan^{-1}\left(\frac{1}{\sqrt{x}}\right)$ is a convex function when $x > 0$, we first define a variable $\tilde U_{m,i}[n]$ which is enforced to satisfy the following constraint:
\begin{align}
    \tilde U_{m,i}[n]\leq\left\|\textbf q_m[n]-\textbf g_i\right\|_2^2. \label{eq_U_relaxation}
\end{align}
Thus, we can derive
\begin{align}
    \frac{180}{\pi}\tan^{-1}\left(\frac{H}{\sqrt{\tilde U_{m,i}[n]}}\right)\geq\frac{180}{\pi}\tan^{-1}\left(\frac{H}{\left\|\textbf q_m[n]-\textbf g_i\right\|_2}\right). \label{eq_tan_U_q_relaxation}
\end{align}
By using (\ref{eq_tan_U_q_relaxation}), the angle elevation angle constraint (\ref{eq_imposed_angle_elevation}) can be replaced by an upper bound constraint:
\begin{align}
    &\theta_{m,i}[n]\geq\frac{180}{\pi}\tan^{-1}\left(\frac{H}{\sqrt{\tilde U_{m,i}[n]}}\right),\forall n \in\{ 1,\ldots,N-1\},\nonumber\\
    &\;\;\;\;\;\;\forall m\in\{1\ldots,M\},\forall i\in\{1,\ldots,K\}.\label{eq50}
\end{align}
Moreover, the constraint (\ref{eq_U_relaxation}) is non-convex in terms of $\textbf q_m[n]$, and we again apply the first-order Taylor expansion to find a lower bound for $\left\|\textbf q_m[n]-\textbf g_i\right\|_2^2$:
\begin{align}
    \left\|\textbf q_m[n]-\textbf g_i\right\|_2^2&\geq\left\|\textbf q_m^r[n]-\textbf g_i\right\|_2^2 \nonumber\\
    &\;\;\;+2(\textbf q^r_m[n]-\textbf g_i)^T(\textbf q_m[n]-\textbf q^r_m[n]), \label{eq_lower_bound_q_distance}
\end{align}where $\textbf q^r_m[n]$ is a given reference point. Hence, the constraint (\ref{eq_U_relaxation}) can be convexified through the lower bound relationship in (\ref{eq_lower_bound_q_distance}), given by
\begin{align}
    &\tilde U_{m,i}[n]\leq\left\|\textbf q_m^r[n]-\textbf g_i\right\|_2^2+2(\textbf q^r_m[n]-\textbf g_i)^T(\textbf q_m[n]-\textbf q^r_m[n]),\nonumber\\
    &\;\;\forall n \in\{ 1,\ldots,N-1\},\forall m\in\{1\ldots,M\},\forall i\in\{1,\ldots,K\},\label{eq52}
\end{align}where the solutions that satisfy this constraint are also in the feasible set of the constraint (\ref{eq_U_relaxation}).

In summary, the convex constraints (\ref{eq45}) and (\ref{eq46}) convexify the constraints (\ref{eq6}) and (\ref{eq39add}) in the subproblem (\textbf{P5}), respectively, while the convex constraints (\ref{eq48}), (\ref{eq50}) and (\ref{eq52}) are able to convexify the constraint (\ref{eq40add}). Accordingly, the UAV flight trajectory subproblem (\textbf{P5}) can be rewritten as 
\begin{align}
\textbf{(P6)}\;&\max_{\{\zeta_q,\textbf{q}_m\left[n\right],\tilde X_{m,i}[n],\tilde Y_{m,i}[n],\theta_{m,i}[n],U_{m,i}[n],\forall m,k,n\}}{\zeta_q}\nonumber\\
  s.t.\;
  &(\ref{eq4}),(\ref{eq5}),(\ref{eq14}),(\ref{eq43}),(\ref{eq45}),(\ref{eq46}),(\ref{eq48}),(\ref{eq50}),(\ref{eq52})\nonumber.
\end{align}
The UAV flight trajectory subproblem (\textbf{P6}) now becomes a solvable convex optimization problem under a given reference point $\textbf{q}_m^r\left[n\right]$, and the optimal value of the UAV trajectory can be iteratively obtained by using the existing solution tool (CVX) {\cite{M. Grant16}} with SCA methods {\cite{M. Razaviyayn14}}. Note that the obtained optimal solution of the subproblem (\textbf{P6}) is a lower bound of the subproblem (\textbf{P4}), since the feasible set of (\textbf{P4}) contains that of (\textbf{P6}).

\subsection{Transmit Power Control Subproblem}
Given the communication association $a_{m,k}[n]$ and the UAV flight trajectory $\textbf q_m[n]$, the transmit power control subproblem for WNs can be obtained from the original optimization problem (\textbf{P2}) by using the epigraph form and introducing an auxiliary variable $\zeta_p$, as follows:
\begin{align}
  \textbf{(P7)}\;&\max_{\{\zeta_p,P_k[n],\forall k,n\}}{\zeta_p}\nonumber\\
  s.t.\;
  &\bar R_k\geq\zeta_p,\forall k\in\{1,...,K\},\label{eq53}\\
  &(\ref{eq17}),(\ref{eq18}) .\nonumber
\end{align}
In this subproblem, the objective function and the constraints (\ref{eq17}) and (\ref{eq18}) are affine and thus convex in $P_k[n]$. However, from (\ref{eq31}), it can be found that the user rate $\bar R_k$ is the sum of concave functions and convex functions, and thus $\bar R_k$ in the constraint (\ref{eq53}) is neither convex nor concave in terms of the variables $P_i[n]$. To solve this problem, we convexify the constraint (\ref{eq53}) by the first-order Taylor expansion. From (\ref{eq31}), it is known that $\check R_{2m,k}[n]$ is a convex function in $P_i[n]$, and we can get the following lower bound relationship:
\begin{align}
&    \check R_{2m,k}[n] \geq -{\log}_2\left(\sum_{i=1,i\neq k}^{K}P_i^r[n]\bar H_{m,i}[n]+\sigma_n^2\right)\nonumber\\
                     &\;\;\;-\sum_{i=1,i\neq k}^{K}\frac{\bar H_{m,i}[n]}{\ln(2)\left(\sum_{j=1,j\neq k}^{K}P^r_j[n]\bar H_{m,j}[n]+\sigma_n^2\right)}\nonumber\\
                     &\;\;\;\times\left(P_i[n]-P^r_i[n]\right) \triangleq \acute R_{m,k}^{lb}[n], \forall n \in\{ 1,\ldots,N-1\},\nonumber\\
&   \;\;\;\;\;\; \;\;\;\;\;\; \;\;\;\;\;\;  \forall m\in\{1\ldots,M\},\forall k\in\{1,\ldots,K\}, \label{eq55}
\end{align}
where $P_i^r[n]$ is a reference point for the Taylor expansion. Accordingly, the user rate $\bar R_k$ in (\ref{eq31}), i.e., the left-hand side of the constraint (\ref{eq53}), can be lower bounded by a concave function, as follows:
\begin{align}
& \bar R_k \geq \sum_{n=0}^{N-1}\sum_{m=1}^{M}a_{m,k}[n]W\Biggl(  \check R_{1m,k}[n]
    +\acute R^{lb}_{m,k}[n]\Biggr), \nonumber \\ 
   &\;\;\;\;\;\;\;\;\;\;\;\;\;\;\;\; \forall k\in\{1,\ldots,K\}.\label{eq57}
\end{align}
Thus, the subproblem (\textbf{P7}) can be transformed into a convex one by replacing the constraint (\ref{eq53}) with the lower bound (\ref{eq57}):
\begin{align}
  \textbf{(P8)}\;&\qquad\qquad\qquad\quad\max_{\{\zeta_p,P_k[n],\forall k,n\}}{\zeta_p}\nonumber\\
  s.t.\;&\sum_{n=0}^{N-1}\sum_{m=1}^{M}a_{m,k}[n]W\Biggl(  \check R_{1m,k}[n]  +\acute R^{lb}_{m,k}[n]\Biggr)\geq \zeta_p, \nonumber\\
 & \;\;\;\;\;\;\;\;\;\;\;\;\;\;\;\;\;\;\;\;\;\;\;\;\;\; \forall k\in\{1,\ldots,K\},\label{eq_power_constraint_first_order}\\
        &(\ref{eq17}),(\ref{eq18}).\nonumber
\end{align}
Note that the optimization result of (\textbf{P8}) is a lower bound of (\textbf{P7}). Given a reference point $P_i^r[n]$, the transmit power control subproblem (\textbf{P8}) can be iteratively solved by the optimization tool, e.g., CVX, with SCA methods.

After the above transformation, the offline joint design problem (\textbf{P2}) can be alternatively solved by the three convex subproblems (\textbf{P3}), (\textbf{P6}) and (\textbf{P8}), and the proposed offline algorithm is summarized in Algorithm 1. For convenience, we assume that the superscript $r$ of the reference points $\textbf q_m^r[n]$, $P_k^r[n]$, $a_{m,k}^r[n]$ also refers to the iteration index in the successive convex optimization. We first initialize $r=0$, and initialize $\textbf q_m^r[n]$, $P_k^r[n]$, and $a_{m,k}^r[n]$. Let $\epsilon >0$ be a threshold for the stopping criterion.
In the communication association subproblem (\textbf{P3}), the optimal solution to the communication association $a_{m,k}^{\star}[n]$ is obtained under the given values of UAV trajectory $\textbf q_m^r[n]$ and power control $P_k^r[n]$, and then we update $a_{m,k}^{r+1}[n]= a_{m,k}^{\star}[n]$. Afterwards, the UAV flight trajectory subproblem (\textbf{P6}) is solved under the given values of communication association $a_{m,k}^{r+1}[n]$ and power control $P_k^r[n]$ through the successive convex optimization. In the inner loop of the successive convex optimization, we update $\textbf q_m^r[n]$ with the last result of $\textbf q_m^\star[n]$, and subsequently renew the related constraints that contain $\textbf q_m^r[n]$. Now the new optimal solution $\textbf q_m^\star[n]$ can be obtained by solving the subproblem (\textbf{P6}). The steps in the loop are repeated until the convergence is achieved, and we update $\textbf q_m^{r+1}[n]=\textbf q_m^\star[n]$. Finally, the power control subproblem (\textbf{P8}) is solved under the given values of UAV trajectory $\textbf q_m^{r+1}[n]$ and communication association $a_{m,k}^{r+1}[n]$, where we update $P_k^r[n]$ and (\ref{eq_power_constraint_first_order}) with the last result of $P_k^\star[n]$ in the inner loop of the successive convex optimization. We repeat these steps in the inner loop until the convergence is attained, and update $P_k^{r+1}[n]=P_k^\star[n]$. We then update the iteration number of the outer loop for the next round. The outer loop is stopped when the increase of the objective value is smaller than the preset threshold $\epsilon$.

\begin{algorithm}[t]
\caption{Offline Alternative Optimization Algorithm for Problem (\textbf{P2})}\label{offline_algorithm}
\begin{algorithmic}[1]
\State Initialize {\{$a_{m,k}^r[n], \textbf q_m^r[n], P_k^r[n],\forall m,k,n$\}, $r=0$}
\State Set $\epsilon >0$
\While{Increase of the objective value $<\epsilon$}
    \State \multiline{For given $\{\textbf q_m^r[n],P_k^r[n],\forall m,k,n\},$ find the optimal communication association solution of the problem (\textbf{P3}) as $a_{m,k}^{r+1}[n].$}
    \Statex
    \Statex
    \State For given $\{ a_{m,k}^{r+1}[n],\textbf q_m^r[n], P_k^r[n],\forall m,k,n\}.$
    \Repeat{ (SCA for solving problem (\textbf{P6}))}
        \State Update $\textbf q_m^r[n]$ using the last result $\textbf q_m^\star[n]$.
        \State Update (\ref{eq43}), (\ref{eq45}), (\ref{eq46}), (\ref{eq48}) and (\ref{eq52}) using $\textbf q_m^r[n]$.
        \State \multiline{Find the new optimal solution $\textbf q_m^\star[n]$ by solving the problem (\textbf{P6}).}       
    \Until convergence to the optimal UAV flight trajectory
    \Statey solution $\textbf q_m^\star[n]$
    \State Set $\textbf q_m^{r+1}[n] \gets \textbf q_m^\star[n].$
    \Statex
    \State For given $\{a_{m,k}^{r+1}[n], \textbf q_m^{r+1}[n], P_k^r[n],\forall m,k,n\}.$
        \Repeat{ (SCA for solving problem (\textbf{P8}))}
        \State Update $P_k^r[n]$ using the last result $P_k^\star[n]$.
        \State Update (\ref{eq_power_constraint_first_order}) using $P_k^r[n]$.
        \State \multiline{Find the new optimal solution $P_k^\star[n]$ by solving the problem (\textbf{P8}).} 
    \Until {convergence to the optimal power control strat-}
    \Statey egy solution  $P_k^\star[n]$
    \State Set $P_k^{r+1}[n] \gets P_k^\star[n].$
\State Update $r \gets r+1.$
\EndWhile
\end{algorithmic}
\end{algorithm}

\subsection{Convergence of the Offline Algorithm}\label{Convergence_analysis}
The convergence of the proposed offline algorithm for UAV flight trajectory, WN communication association and power control is analyzed as follows. Let $f\left(a_{m,k}[n],\textbf q_m[n],P_k[n]\right)$ be the objective function of the original problem (\textbf{P2}). First, in the communication association subproblem (\textbf{P3}), by fixing $\left\{\textbf{q}_m^r\left[n\right], P_ k^r\left[n\right]\right\}$ to obtain the optimal solution of the communication association $a_{m,k}^{r+1}\left[n\right]$, it results in
\begin{align}
    &f\left(a_{m,k}^r[n],\textbf{q}_m^r\left[n\right],P_k^r[n]\right)\leq 
    f\left(a_{m,k}^{r+1}[n],\textbf{q}_m^r\left[n\right],P_k^r[n]\right).\label{eq59}
\end{align}

Next we discuss the convergence of the inner loop of the UAV flight trajectory subproblem (\textbf{P6}). Given $\left\{a_{m,k}^{r+1}\left[n\right], \textbf q_m^r[n], P_k^r\left[n\right]\right\}$, the subproblem is optimized by the successive convex optimization to obtain the trajectory $\textbf{q}_m^{r+1}\left[n\right]$. In the $i$th inner loop of the successive convex optimization, we assume that $\textbf{q}_m^{\left(i\right)\star}\left[n\right]$ is the obtained solution with respect to the given first-order Taylor expansion reference point $\textbf{q}_m^{\left(i\right)}\left[n\right]$. Further, denote $\zeta_{\textbf{q}_m^{\left(i\right)}\left[n\right]}\left(a_{m,k}^{r+1}[n],\textbf{q}_m^{\left(i\right)\star}\left[n\right],P_k^r[n]\right)$ as the worst sum rate for the obtained solution $\textbf{q}_m^{\left(i\right)\star}\left[n\right]$ at the reference point $\textbf{q}_m^{\left(i\right)}\left[n\right]$ in the subproblem (\textbf{P6}). Then we have 
\begin{align}
& \zeta_{\textbf{q}_m^{\left(i\right)}\left[n\right]}\left(a_{m,k}^{r+1}[n],\textbf{q}_m^{\left(i\right)\star}\left[n\right],P_k^r[n]\right)\nonumber\\
& \leq f\left(a_{m,k}^{r+1}[n],\textbf{q}_m^{\left(i\right)\star}\left[n\right],P_k^r[n]\right) \nonumber\\
& = \zeta_{\textbf{q}_m^{\left(i\right)\star}\left[n\right]}\left(a_{m,k}^{r+1}[n],\textbf{q}_m^{\left(i\right)\star}\left[n\right],P_k^r[n]\right) \nonumber\\
& \leq  \zeta_{\textbf{q}_m^{\left(i+1\right)}\left[n\right]}\left(a_{m,k}^{r+1}[n],\textbf{q}_m^{\left(i+1\right)\star}\left[n\right],P_k^r[n]\right) ,  \label{eq59_new}
\end{align}where the first inequality comes from the lower bound relationship in (\ref{eq42}), and the second equality is because the equality of the lower bound relationship holds at the tangent point. Moreover, the third inequality is due to the fact that we set $\textbf{q}_m^{\left(i+1\right)}\left[n\right]= \textbf{q}_m^{\left(i\right)\star}\left[n\right]$ to update the reference point for the next inner iteration according to the SCA and obtain the corresponding optimal solution $\textbf{q}_m^{\left(i+1\right)\star}\left[n\right]$. Hence it concludes that the worst sum rate performance can be monotonically increased in the inner loop of the UAV flight trajectory subproblem (\textbf{P6}), and we can get the following relationship for UAV flight trajectory in the $r$th outer loop:
\begin{align}
& f\left(a_{m,k}^{r+1}[n],\textbf{q}_m^{r}\left[n\right],P_k^r[n]\right) \leq  f\left(a_{m,k}^{r+1}[n],\textbf{q}_m^{r+1}\left[n\right],P_k^r[n]\right) ,  \label{eq59_new2}
\end{align}

Likewise, we can apply the similar derivation of (\ref{eq59_new}) to prove that the sum rate performance can be monotonically increased in the inner loop of the SCA for the power control subproblem (\textbf{P8}). Thus, under the given $\left\{a_{m,k}^{r+1}\left[n\right], \textbf q_m^{r+1}[n], P_k^r\left[n\right]\right\}$, it implies that 
\begin{align}
& \zeta_{P_k^{\left(i\right)}[n]}\left(a_{m,k}^{r+1}[n],\textbf{q}_m^{r+1}\left[n\right],P_k^{ \left(i\right)\star}[n]\right) \nonumber\\
&\leq  \zeta_{P_k^{\left(i+1\right)}[n]}\left(a_{m,k}^{r+1}[n],\textbf{q}_m^{r+1}\left[n\right],P_k^{\left(i+1\right)\star}[n]\right) ,  \label{eq59_new3}
\end{align}where $\zeta_{P_k^{\left(i\right)}[n]}\left(a_{m,k}^{r+1}[n],\textbf{q}_m^{r+1}\left[n\right],P_k^{ \left(i\right)\star}[n]\right)$ is referred to as the worst sum rate for the obtained solution $P_k^{ \left(i\right)\star}[n]$ at the reference point $P_k^{\left(i\right)}[n]$ in the $i$th inner loop of the power control subproblem (\textbf{P8}). As a result, it implies that
\begin{align}
& f\left(a_{m,k}^{r+1}[n],\textbf{q}_m^{r+1}\left[n\right],P_k^r[n]\right) \leq \nonumber \\
& \;\;\;\;\;\;\;\;\;\;\;\;\;\;\;\;\;\;\; f\left(a_{m,k}^{r+1}[n],\textbf{q}_m^{r+1}\left[n\right],P_k^{r+1}[n]\right) . \label{eq59_new4}
\end{align}
Due to the alternative optimization for the three subproblems, it can be concluded from (\ref{eq59}), (\ref{eq59_new2}) and (\ref{eq59_new4}) that the performance of the proposed algorithm can be monotonically increased and the local optimal solution of the original problem (\textbf{P2}) can be found until the algorithm is converged.

\section{Convex-Assisted Reinforcement Learning (CARL)}\label{Online}
In this section, we propose a CARL approach, in which a flight corridor is marked out, based on the proposed offline design in Algorithm 1, to guide the UAV flight actions in an online learning fashion. Besides, the actions of the RL agents can be restricted and the number of system states can be efficiently reduced through the assistance of offline design. The system states, actions and real-time rewards are designed and presented in the following.

\subsection{System States}
Let ${\textbf S}= {\textbf L} \times {\textbf H}$ be a two-tuple state space, where $\times$ denotes the Cartesian product, ${\textbf L}$ represents a UAV location state set, and ${\textbf H}$ is a channel state set. Moreover, we define a random variable ${\textbf s}= \left({\textbf l}, {\textbf h} \right) \in {\textbf S}$ as the system stochastic state of the Markov decision process. It is assumed that the UAV location and channel remain steady during the time interval $\delta_D$. The detailed definition of each state is specified in the following.

$\bullet\;\textbf{UAV location state}$: Assume that a square UAV coverage region is quantized into $N_L$ lattice points with a scale of $\Delta$ (the minimum distance between two adjacent horizontal/vertical lattice points), and the UAV location state space is defined as ${\textbf L}= {\textbf L}_x \times {\textbf L}_y$, where ${\textbf L}_x= {\textbf L}_y =\left\{0,1,\ldots,\sqrt{N_L}-1\right\}$. When the location state of the $m$th UAV at the time instant $n$ is given as ${\textbf l}_m^n=\left[{\widetilde{x}}_m^n,{\widetilde{y}}_m^n\right] \in {\textbf L}$, it means that the horizontal coordinate of the $m$th UAV at the time instant $n$ is
\begin{align}\label{eq_q_coordinate}
   {\textbf q}_m\left[n\right]=\left[\Delta\times\widetilde{x}_m^n+\frac{\Delta}{2},\Delta\times\widetilde{y}_m^n+\frac{\Delta}{2}\right]^T. 
\end{align}

Let ${{\textbf q}}^\star_m[n]$ be the offline horizontal flight path of the $m$th UAV at the time instant $n$, obtained by the offline convex optimization. By using the offline UAV flight trajectory ${{\textbf q}}^\star_m[n]$, the offline trajectory-assisted location state of $m$th UAV at the time instant $n$ is given as
\begin{align}
    {\textbf l}_m^n\in\hat{{\textbf L}}_m=\{{\textbf l}_m^n\in{\textbf L} | \|{\textbf q}_m[n]-{{\textbf q}}^\star_m[n]\|_2\leq D_F,\forall n\},\label{carl_uav_location}
\end{align}
where $D_F$ is the flight corridor width, and the distance between the real UAV location ${\textbf q}_m[n]$ and the offline UAV trajectory ${{\textbf q}}^\star_m[n]$ at the time instant $n$ is smaller than the preset corridor width $D_F$. Then the location state of all UAVs at the time instant $n$ can be expressed as
\begin{align}
    {\textbf l}^n=[{\textbf l}_1^n,\ldots,{\textbf l}_M^n]\in \hat{{\textbf L}}_1\times\hat{{\textbf L}}_2\times\cdots\times\hat{{\textbf L}}_M.\label{eq271}
\end{align}

$\bullet\;\textbf{Channel state}$: With the assistance of the offline UAV flight trajectory results, the average LOS/NLOS channel strength ${\bar{H}}_{m,k}\left[n\right]$ between the $k$th WN and the $m$th UAV along the flight path at each time can be calculated by (\ref{average_H}) to simplify the quantization of channel states. Let $\epsilon_H>0$ be a threshold value for channel quantization. Considering the fact that the channel strength in the vicinity of the flight corridor is related to the channel strength of the offline flight path, the channel state of the $m$th UAV can be defined as:
\begin{align}
h_{m,k}^n=\left\{
\begin{aligned}
     &0,H_{m,k}[n]<\bar{H}_{m,k}[n]-\epsilon_H;\\
     &1,\bar{H}_{m,k}[n]-\epsilon_H\leq H_{m,k}[n]\leq \bar{H}_{m,k}[n]+\epsilon_H;\\
     &2,H_{m,k}[n]>\bar{H}_{m,k}[n]+\epsilon_H.\label{CARL_channel_state}
\end{aligned}
\right.
\end{align}
By using this relative quantization method, the changes in channel strength can be better described than the direct quantization of channel strength, and the number of quantization levels can be greatly reduced. As such, the channel state of all UAVs and nodes at the time instant $n$ is defined as:
\begin{align}
    {\textbf h}^n=\left[h_{1,1}^n,\ldots,h_{1,K}^n,\ldots,h_{M,1}^n,\ldots,h_{M,K}^n\right]\in{\hat{\textbf H}}^{M\times K},
\end{align}where ${\hat{\textbf H}}=\left\{0, 1, 2\right\}$.

\subsection{System Actions}
Based on the states of UAV locations and channels, we can decide the UAV flight directions, the association between UAVs and WNs, and uplink transmit power. Define ${\textbf A}= {\textbf A}_F\times {\textbf A}_C$ as the action space, where ${\textbf A}_F=\left\{0,1,2,3,4\right\}$ represents the set of five flight direction actions, ${\textbf A}_C=\left\{0,1,\ldots, N_pK \right\}$ is the set of communication (including association and transmission) actions, and $N_p$ is the number of transmit power levels available for each WN. Denote $a_{m,F}^n$ and $a_{m,C}^n$ as the UAV flight direction action and the communication action of the $m$th UAV at the time instant $n$, respectively. Hence, the concatenated action of all UAVs that can be chosen at the time instant $n$ is given as ${\textbf a}^n=\left[ {\textbf a}_F^n,{\textbf a}_C^n\right]$, where ${\textbf a}_F^n= \left[a_{1,F}^n, \ldots ,a_{M,F}^n\right]$ and ${\textbf a}_C^n= \left[a_{1,C}^n, \ldots ,a_{M,C}^n\right]$. The details of the actions are specified below. For the flight direction action ${\textbf a}_F^n \in {\textbf A}_F^M$, the action $a_{m,F}^n=0$ means that the $m$th UAV is hovering at the time instant $n$, while the action $a_{m,F}^n=1, 2, 3, 4$ means that the $m$th UAV moves to the left, right, forward or backward with a predefined distance $\Delta$, respectively. By flight corridor restrictions, for a given system state ${\textbf s}^n=\left[{\textbf l}^n,{\textbf h}^n\right]$ at the time instant $n$, the UAV flight action is partially constrained by ${\textbf a}_F^n \in{\hat{{\textbf A}}}_{F}^n=\{ {\textbf a}_F^n \in {\textbf A}_F^M \mid \|{\textbf q}_m[n+1]-{{\textbf q}}^\star_m[n+1]\|_2\leq D_F,\forall m, {\text{and}}\ {\textbf q}_m[n+1]\neq {\textbf q}_j[n+1],\forall m\neq j\}$. In this setting, the actions that can ensure that the UAVs' position at the next time being within the range of the flight path and avoid the collision among the multiple UAVs are regarded as legal actions.

Let ${\bar{\textbf A}}_C^n$ be the affordable communication action set that satisfies the battery constraints and the communication association constraints at the time instant $n$, i.e., ${\textbf a}_C^n \in {\bar{\textbf A}}_C^n \subseteq {\textbf A}_C^M$, and it can be constructed as follows. Assume that $p_k^n$ and $b_k^n$ are the transmit power level and battery power of the $k$th WN at the time instant $n$, respectively, and $p_k^n \in {\textbf P}=\{{1,2,...,N_p}\}$. For $p_k^n$, it means that the energy expenditure of the $k$th WN for uplink data transmission is $p_k^n E_p$ at the time instant $n$, for $k=1,\ldots,K$ and $p_k^n=1,\ldots,N_p$, where $E_p$ is the basic energy unit with respect to the transmit power level $p_k^n=1$. If the $k$th WN is associated with the $m$th UAV with the uplink transmit power level $p_k^n$ at the time instant $n$, the action $a_{m,C}^n=(k-1)N_p+p_k^n$ is performed. On the other hand, the action $a_{m,C}^n=0$ indicates that the $m$th UAV is not connected to any WNs at the time $n$. In addition, the communication action $a_{m,C}^n= (k-1)N_p+p_k^n$ is constrained by the available battery power of the WNs and the UAV-WN association. For the battery power constraint, the action $a_{m,C}^n= (k-1)N_p+p_k^n$ is eligible to be performed only if the energy expenditure $p_k^n E_p$ is less than the battery power $b_k^n$ of the $k$th WN, i.e., $p_k^n E_p\leq b_k^n$. For the UAV-WN association, if $a_{m,C}^n \neq 0$, the actions of the other UAVs are limited by $a_{j,C}^n \neq a_{m,C}^n$, for any $j\neq m$. This is because each WN can be only served by one UAV during a time interval.
{
\subsection{State Transition}
After all UAVs perform an action ${\textbf a}^n$ at the time $n$, the current system state ${\textbf s}^n$ is transited to the next system state ${\textbf s}^{n+1}$, and the state transitions are elaborated as follows.

$\bullet\;\textbf{UAV location state}$: Assume the current UAV location state is ${\textbf l}_{m}^{n}=[\tilde x_m^{n},\tilde y_m^{n}]$, and the action ${\textbf a}_m^n= \left[a_{m,F}^n, a_{m,C}^n\right]$ is performed. The next UAV location state becomes
\begin{align}
{\textbf l}_{m}^{n+1}=[\tilde x_m^{n+1},\tilde y_m^{n+1}]=\left\{
\begin{aligned}
     &[\tilde x_m^n,\tilde y_m^n],\;a_{m,F}^n=0;\\
     &[\tilde x_m^n-1,\tilde y_m^n],\;a_{m,F}^n=1;\\
     &[\tilde x_m^n+1,\tilde y_m^n],\;a_{m,F}^n=2;\\
     &[\tilde x_m^n,\tilde y_m^n+1],\;a_{m,F}^n=3;\\
     &[\tilde x_m^n,\tilde y_m^n-1],\;a_{m,F}^n=4.\label{eq24}
\end{aligned}
\right.
\end{align}
From (\ref{eq3}) and (\ref{eq_q_coordinate}), the corresponding change of the m$th$ UAV coordinates ${\textbf q}_{m}\left[n\right]$ with respect to the action ${\textbf a}_m^n= \left[a_{m,F}^n, a_{m,C}^n\right]$ is thus given by
\begin{align}
{\textbf q}_{m}\left[n+1\right]=\left\{
\begin{aligned}
     &[x_m\left[n\right],y_m\left[n\right]]^T,\;a_{m,F}^n=0;\\
     &[x_m\left[n\right]-\Delta,y_m\left[n\right]]^T,\;a_{m,F}^n=1;\\
     &[x_m\left[n\right]+\Delta,y_m\left[n\right]]^T,\;a_{m,F}^n=2;\\
     &[x_m\left[n\right],y_m\left[n\right]+\Delta]^T,\;a_{m,F}^n=3;\\
     &[x_m\left[n\right],y_m\left[n\right]-\Delta]^T,\;a_{m,F}^n=4.\label{eq23}
\end{aligned}
\right.
\end{align}
Note that if the UAV selects the communication action, its coordinates and location state remain unchanged, i.e., $a_{m,F}^n=0$ if $a_{m,C}^n \neq 0$.

$\bullet\;\textbf{Channel state}$ : The channel state $h_{m,k}^{n+1}$ between the $m$th UAV and the $k$th WN at the time instant $(n+1)$ is decided by quantizing the observed channel gain $H_{m,k}[n+1]$ according to (\ref{CARL_channel_state}), where the instantaneous channel $H_{m,k}[n+1]$ depends on the positions of the UAV and the WN. Based on the channel model in (\ref{eq7})--(\ref{eq10}), the large-scale channel components of $H_{m,k}[n+1]$ and $H_{m,k}[n]$ are correlated with each other.
}
\subsection{Reward Design}

During the entire mission time $n=0, \ldots, N-1$, the system receives a negative reward $C_P<0$ as a penalty if the UAV flight action set $\hat{{\textbf A}}_{F}^n$ constrained by the offline UAV trajectory is empty, i.e., the next position of any UAV is beyond the flight corridor. On the contrary, if the UAV flight action set $\hat{{\textbf A}}_{F}^n$ is non-empty, we discuss the reward in two cases: 1) $n=0,\ldots,N-2$ and 2) $n=N-1$. For $n=0,\ldots,N-2$, the system can get a reward $\rho({\textbf s}^n, {\textbf a}^n)$ at a state ${\textbf s}^n$ with respect to a performed action ${\textbf a}^n \in \hat{{\textbf A}}_{F}^n$. At $n=N-1$, all the UAVs are required to return to the starting point when the action is performed. Hence, the system gets a negative reward $C_P<0$ if any UAV does not comply with this constraint, i.e., ${\textbf l}_m^N\neq {\textbf l}_m^0$ for any $m$; otherwise, the system can get a reward $\rho({\textbf s}^n, {\textbf a}^n)$ when all the UAVs can successfully flight back to the starting point. The system reward can be summarized as
\begin{align}\label{Reward}
&R^n({\textbf s}^n,{\textbf a}^n)= \nonumber\\
&\left\{
\begin{aligned}
& C_P, n=\{0,...,N-1\}, {\text{and}} \ \hat{{\textbf A}}_{F}^n= \varnothing;\\    
    & \rho({\textbf s}^n, {\textbf a}^n), n=\{0,...,N-2\}, {\text{and}} \ \hat{{\textbf A}}_{F}^n\neq \varnothing;\\
    & C_P, n=N-1, \hat{{\textbf A}}_{F}^n\neq \varnothing, {\text{and}} \ {\textbf l}_m^N \neq {\textbf l}_m^0 \;{\text{for any}} \ m; \\
    & \rho({\textbf s}^n, {\textbf a}^n), n=N-1, \hat{{\textbf A}}_{F}^n\neq \varnothing, {\text{and}} \ {\textbf l}_m^N= {\textbf l}_m^0, \forall m.
\end{aligned}
\right.
\end{align}
The design of the reward $\rho({\textbf s}^n, {\textbf a}^n)$ is related to the goal of the original design problem \textbf{(P1)} which attempts to maximize the worst sum rate of WNs. Here, we consider three kinds of reward designs for $\rho({\textbf s}^n, {\textbf a}^n)$, namely worst accumulated sum rate among users (WASR), difference of worst accumulated sum rate among users at two adjacent time slots (DWASR) \cite{U. F. Siddiqi20}, and instantaneous average sum rate of users (ISR), which are in turn defined as follows:
\begin{align}
    &  \rho_{_{WASR}}({\textbf s}^n, {\textbf a}^n)={\min_{k=1,\ldots,K}}\Biggl(Z_k^{n-1}+R_{k,n}\Biggr);\label{WASR}\\
    & \rho_{_{DWASR}}({\textbf s}^n, {\textbf a}^n)={\min_{k=1,\ldots,K}}\Biggl(Z_k^{n-1}+R_{k,n}\Biggr) \nonumber\\
    & \;\;\;\;\;\;\;\;\;\;\;\;\;\;\;\;\;\;\;\;\;\;\;\;\;\;\; \;\;\;\;\;\;\;\;\;\;\; -\min_{k=1,\ldots,K}Z_k^{n-1};\label{D_WASR}\\
    & \rho_{_{ISR}}({\textbf s}^n, {\textbf a}^n)=\frac{1}{K}\sum_{k=1}^{K} R_{k,n},\label{ISR}
\end{align}
where $Z_k^{n-1}$ represents the accumulated sum data rate for the $k$th WN up to the $(n-1)$th time, and $R_{k,n}$ is calculated by (\ref{eq16}) with respect to the action ${\textbf a}^n$ and the state ${\textbf s}^n$ at the time instant $n$. The idea of the WASR method is to simply use the worst accumulated sum rate of the WNs as the reward according to the objective function of the optimization problem (\textbf{P1}). The DWASR method is conceptualized in accordance with \cite{U. F. Siddiqi20}, for which the difference of the worst accumulated sum rate at two adjacent time slots $n$ and $n-1$ is computed as the reward. The ISR method is proposed in this paper to take the instantaneous average sum rate of all WNs as the reward. 

In Algorithm 2, we summarize the procedures of the proposed CARL algorithm. The algorithm starts from the offline UAV flight trajectory ${\textbf q}^\star_m[n]$ obtained by using Algorithm \ref{offline_algorithm}. Let $N_e$ be the number of episodes. We develop the flight corridor based on ${\textbf q}^\star_m[n]$ and find the legal UAV flight actions ${\hat{{\textbf A}}}_F^n$ under the flight corridor constraint and the legal communication actions ${\bar{{\textbf A}}}_{C}^n$ under the node battery constraint. If the set of legal UAV flight actions ${\hat{{\textbf A}}}_F^n$ is non-empty, action selection is performed based on a decaying $\epsilon$-greedy approach {\cite{H. Afifi20}}, which means that $\epsilon$ decreases slowly as training proceeds. After performing an action, the corresponding reward is received according to (\ref{Reward}), which is used to update the Q-table. Note that if the set of legal UAV flight actions ${\hat{{\textbf A}}}_F^{n+1}$ at the next time instant $n+1$ is empty or the mission reaches to the final time instant $n=N-1$, the current episode round will be terminated after performing the action, and the above procedures are repeated for the next episode round.

\begin{algorithm}[t]
\caption{Convex-assisted Reinforcement Learning (CARL) Algorithm}
\begin{algorithmic}[1]
\State Obtain the UAVs' trajectory ${\textbf q}^\star_m[n]$ through the offline convex optimization in Algorithm 1.
\State Initialize Q-value $Q({\textbf s}, {\textbf a})=0$ for all ${\textbf s}$ and ${\textbf a}$.
\For {$j=0$ to $N_e$ ($N_e$ is the number of episodes)}
    \State Set $n=0$
    \While{$n<N$}
        \State \multiline{Find available action sets ${\hat{{\textbf A}}}_F^n$ and ${\bar{{\textbf A}}}_{C}^n$ that meet the flight corridor limits as well as the battery constraints.\\}
        \If{${\hat{{\textbf A}}}_F^n \neq \varnothing$} 
             \State Draw a random number $\epsilon_0$ between 0 and 1.
             \If{$\epsilon_0 < \epsilon$}
               \State \multiline{Perform an action ${\textbf a}^n= ({\textbf a}_F^n, {\textbf a}_C^n)\in {\hat{{\textbf A}}}_F^n \times {\bar{{\textbf A}}}_{C}^n$ randomly.}
            \Else
               \State \multiline{Perform the action ${\textbf a}^n$ with the highest Q-value, i.e., ${\textbf a}^n=$ $\arg\max_{{\textbf a}}$ $Q({\textbf s}^n, {\textbf a})$.}
            \EndIf
        \EndIf
        \If{\;(${\hat{{\textbf A}}}_F^{n+1} = \varnothing$ or $n=N-1$) }
            \State \multiline{Update Q-table as $Q(\textbf s^n,\textbf a^n)\gets(1-a)\cdot Q(\textbf s^n,\textbf a^n)+a \cdot R^n({\textbf s}^n,{\textbf a}^n)$}    
            \State $\textbf{break};$
       \Else
           \State \multiline{Update Q-table as \\$Q(\textbf s^n,\textbf a^n)\gets(1-a)\cdot Q(\textbf s^n,\textbf a^n)+a\cdot \left(R^n({\textbf s}^n,{\textbf a}^n)+\gamma\max\limits_{{\textbf a}}Q(\textbf s^{n+1},\textbf a)\right)$}
           \State Update $n\gets n+1$
       \EndIf
    \EndWhile
\EndFor
\end{algorithmic}
\end{algorithm}
\section{Numerical Simulation}\label{Simulation}
\subsection{Simulation Settings}
For simulation settings, we consider two UAVs flying over a $600$ m $\times$ $600$ m area at an altitude of $150$ m with a mission period time of $T_s= 100$ minutes, where the total number of time slots is set to 100. The initial positions of the first and the second UAVs are $[0,300,150]$ m and $[600,300,0]$ m, respectively. The maximum speed limit of the UAVs is $1$ m/sec, and the safety distance between any two UAVs is $100$ m. The system has a carrier frequency of $2.4$ GHz with a transmission bandwidth of $5$ MHz. The solar panel size of the wireless node is $10$ cm $\times$ $10$ cm. The solar EH data is obtained from NREL website \cite{NREL}, and we use 16 years of monitoring data at Elizabeth State University from $1997$ to $2012$, in which the average solar irradiance profiles over days in the morning (from $10:40$ am to $12:20$ pm and afternoon (from $12:20$ pm to $14:00$ pm) are depicted in Fig. \ref{fig_EH_profile}. The environmental parameters for the channel models are set to $A= 9.61$, $B= 0.1592$, $\eta_{LOS}= 1$ dB, and $\eta_{NLOS}= 20$ dB \cite{J. Holis08}. The stopping criterion in the proposed offline method is given as $\epsilon= 10^{-4}$. For the CARL, the parameters for the quantization of system states are given by $\Delta= 60$ m and $\epsilon_H= 5$ dB. Moreover, we set the discount factor $\gamma=0.5$ and the penalty $C_P=-10^3$, and the decaying learning rate and decaying $\epsilon$-greedy are adopted with $[a_{\max},a_{\min}]=[0.9,0.3]$ and $[\epsilon_{\max},\epsilon_{\min}]=[0.9,0.1]$\footnote{The learning rate is given as $a=(a_{\max}-a_{\min})\times \max(\frac{N_e-n_{step}}{N_e},0)+a_{\min}$, where $n_{step}$ is the number of learning steps so far. The same method is applied for the decaying $\epsilon$-greedy.}. The width of the flight corridor is set to $D_F=90$ m in the morning and $D_F=105$ m in the afternoon. The values of the battery capacity $B_{\max}$, noise power $\sigma_n^2$, training number $N_e$ are set to $1500$ J, $-80$ dBm and $2\times10^6$, respectively. The above parameters are used as default settings, except as otherwise stated.

For the offline designs, some heuristic methods are considered for performance comparison. An exhaustive power control method {\cite{M. Ku15}} is applied, in which each WN exhausts the available energy at each time slot for transmit power control. For the UAV trajectory, three flight plans are considered for the two UAVs: (1) uncrossed circles (UC), (2) crossed circles (CC), (3) straight lines and circles (SLC), as shown in Fig. \ref{fig_heuristic_fly}. For the communication association, the nearest user association method is applied {\cite{R. Amer20}}, and the color of the line in Fig. \ref{fig_heuristic_fly} represents the closest association node. 

For the online design, a conventional RL method is included for performance comparison, in which the UAV location state is defined as in (\ref{eq_q_coordinate}) and the channel state is quantized into three states with the threshold $\left[-100, -90\right]$ (dB). For the conventional RL, the system receives a reward $R^n({\textbf s}^n,{\textbf a}^n)$:
\begin{align}\label{Convential_Reward}
&R^n({\textbf s}^n,{\textbf a}^n)= \nonumber\\
&\left\{
\begin{aligned}
    & f\left( R_{k,n}\right)-\mu_n\bigg(\sum_{m=1}^M\|{\textbf l}_{m}^{n+1}-{\textbf l}_{m}^0\|_2\bigg),  n=\{0,...,N-1\};\\
    & C_P,\; n=N-1, {\textbf l}_m^N \neq {\textbf l}_m^0 \;{\text{for any}} \ m, \\
\end{aligned}
\right.
\end{align}
where {$f\left( R_{k,n}\right)$ is a function of user rates $R_{k,n}$}, the coefficient $\mu_n=10^{-4}n$ is positive and increased with the time index $n$, and it is applied before the total distance metric to force the UAVs to learn to return the starting position at the end of the mission. Note that at $n=N-1$, all the UAVs must comply with the constraint of returning to the starting point at the end of the mission; otherwise, the system receives a negative penalty $C_P=-10^3$.

\begin{figure}[htbp]
\centering
\includegraphics[width=0.48\textwidth]{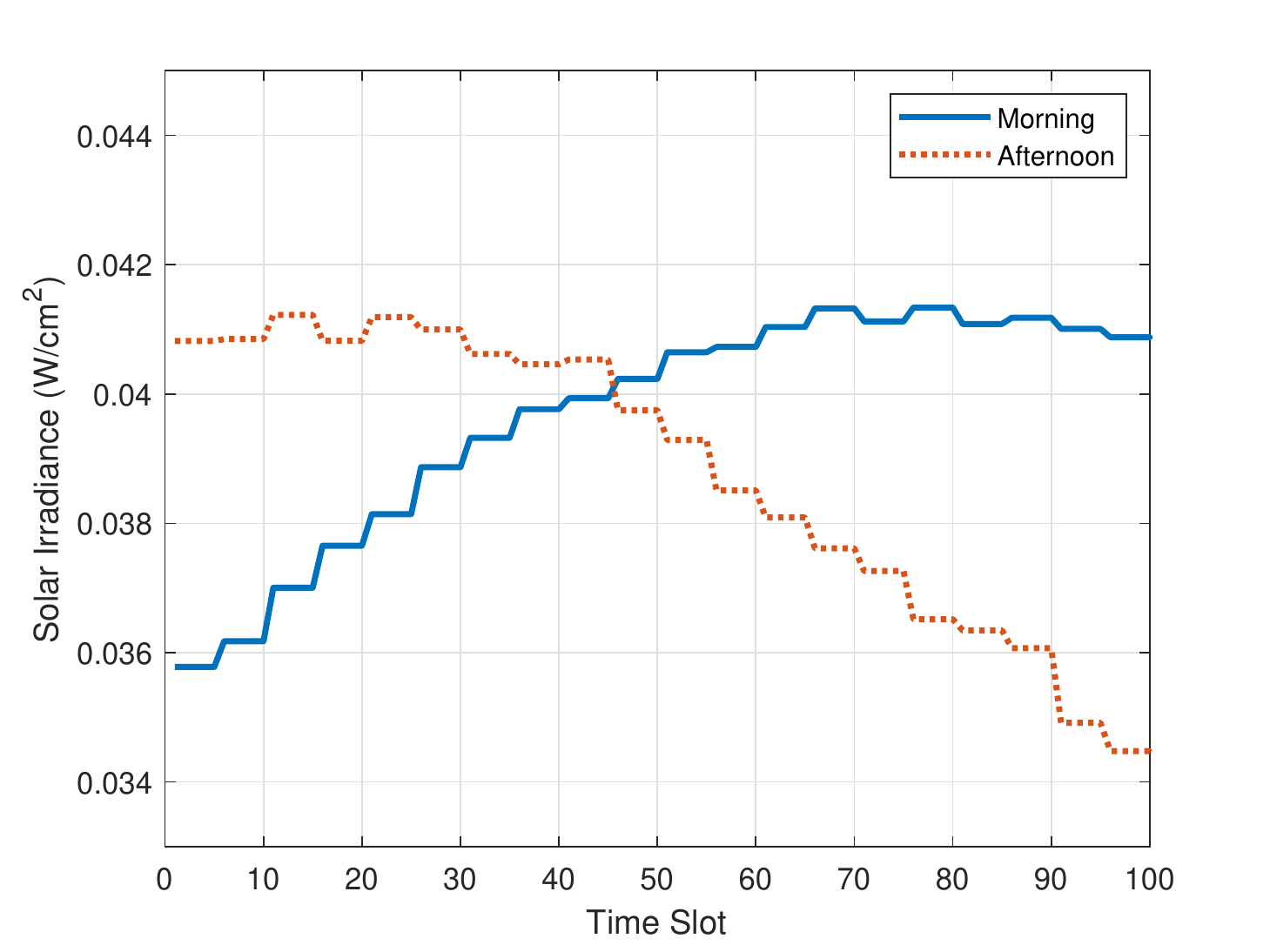}
\caption{Average solar irradiance profiles in the morning and afternoon at the solar site of Elizabeth State University from $1997$ to $2012$.}
\label{fig_EH_profile}
\end{figure}

\begin{figure*}[htbp]
\centering 
\subfigure[{Uncrossed circles (UC)}]{
\label{Fig.sub.1}
\includegraphics[width=0.32\linewidth]{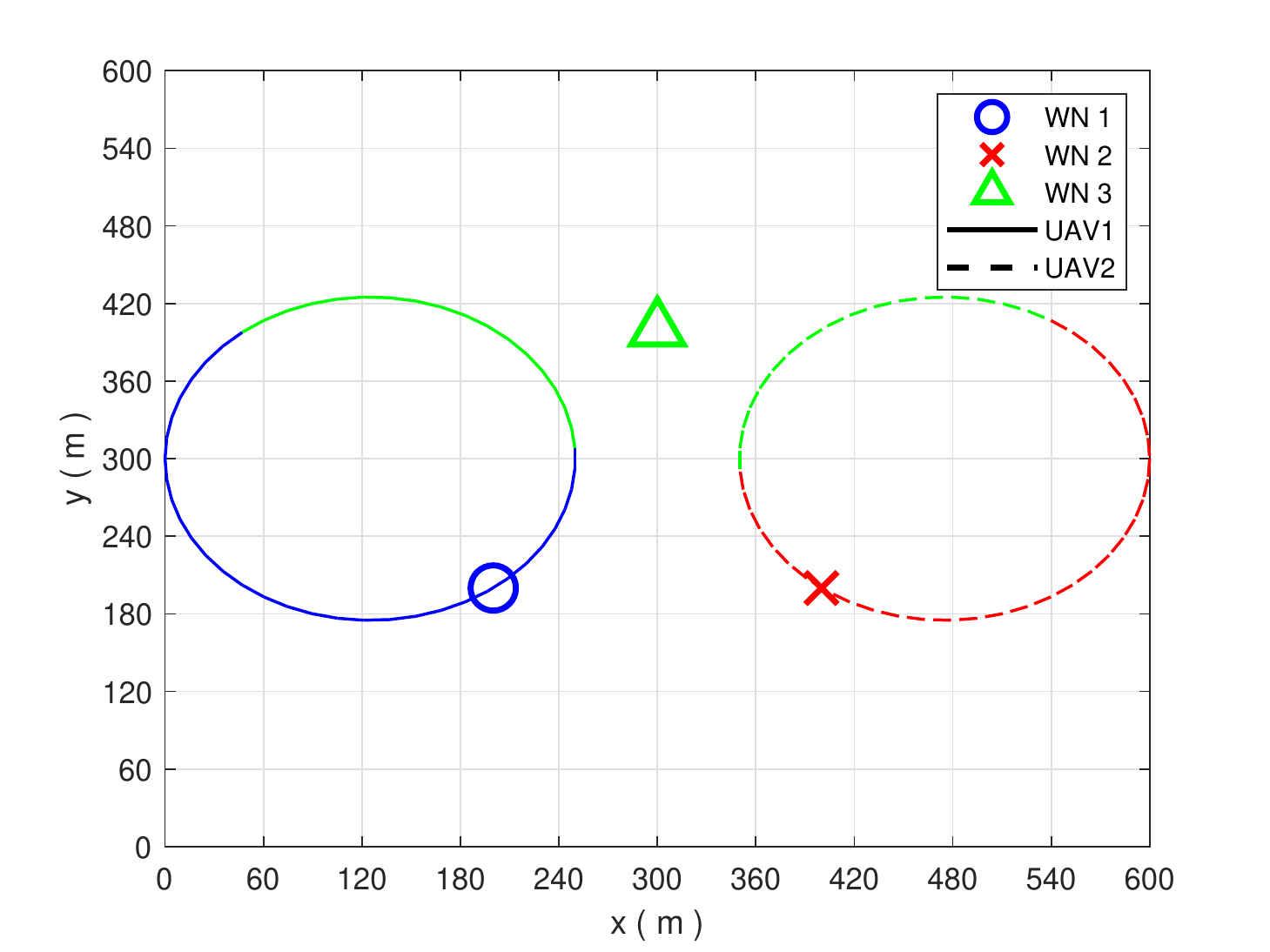}}
\subfigure[{Crossed circles (CC)}]{
\label{Fig.sub.2}
\includegraphics[width=0.32\linewidth]{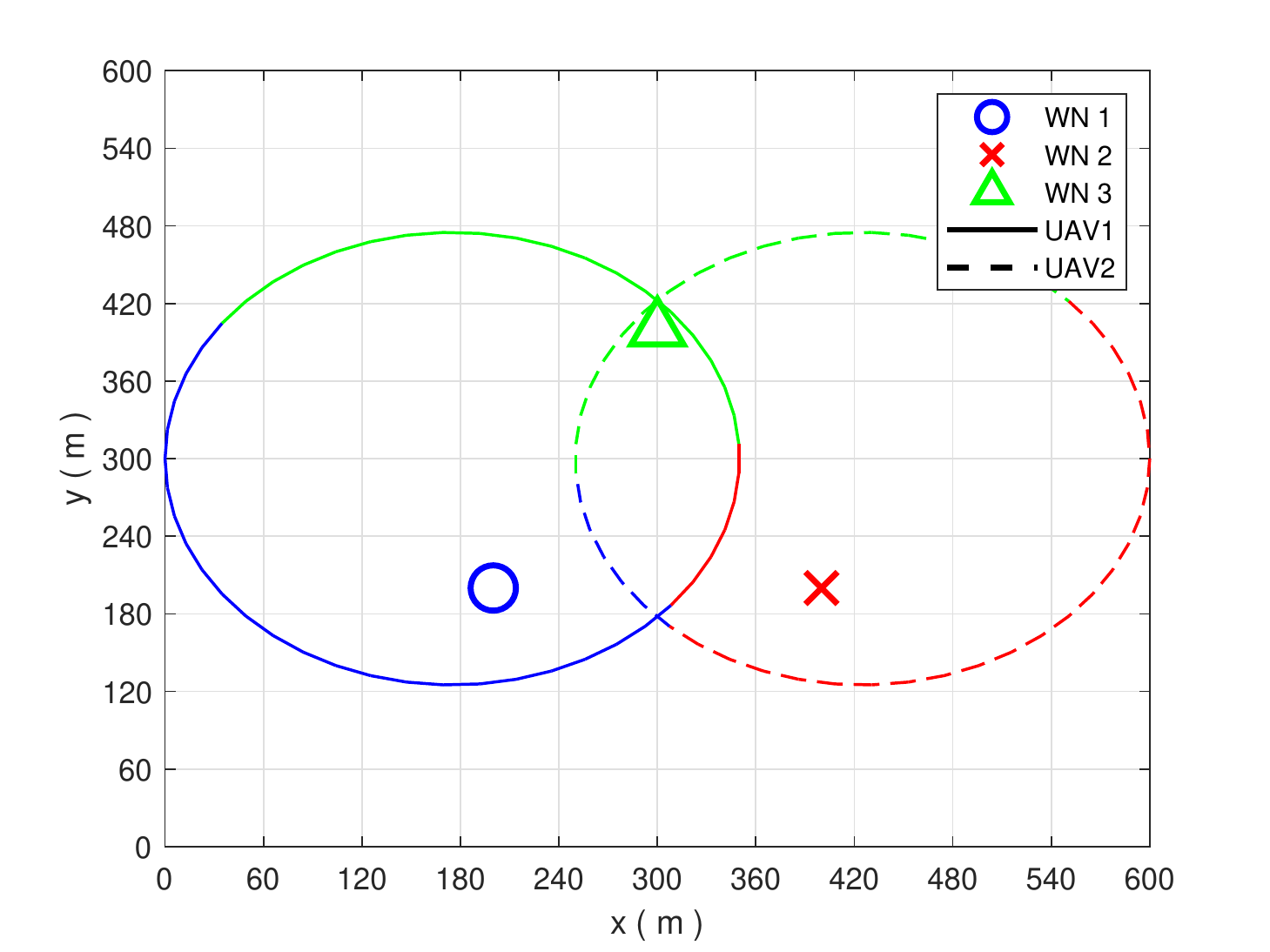}}
\subfigure[{Straight lines and circles (SLC)}]{
\label{Fig.sub.3}
\includegraphics[width=0.32\linewidth]{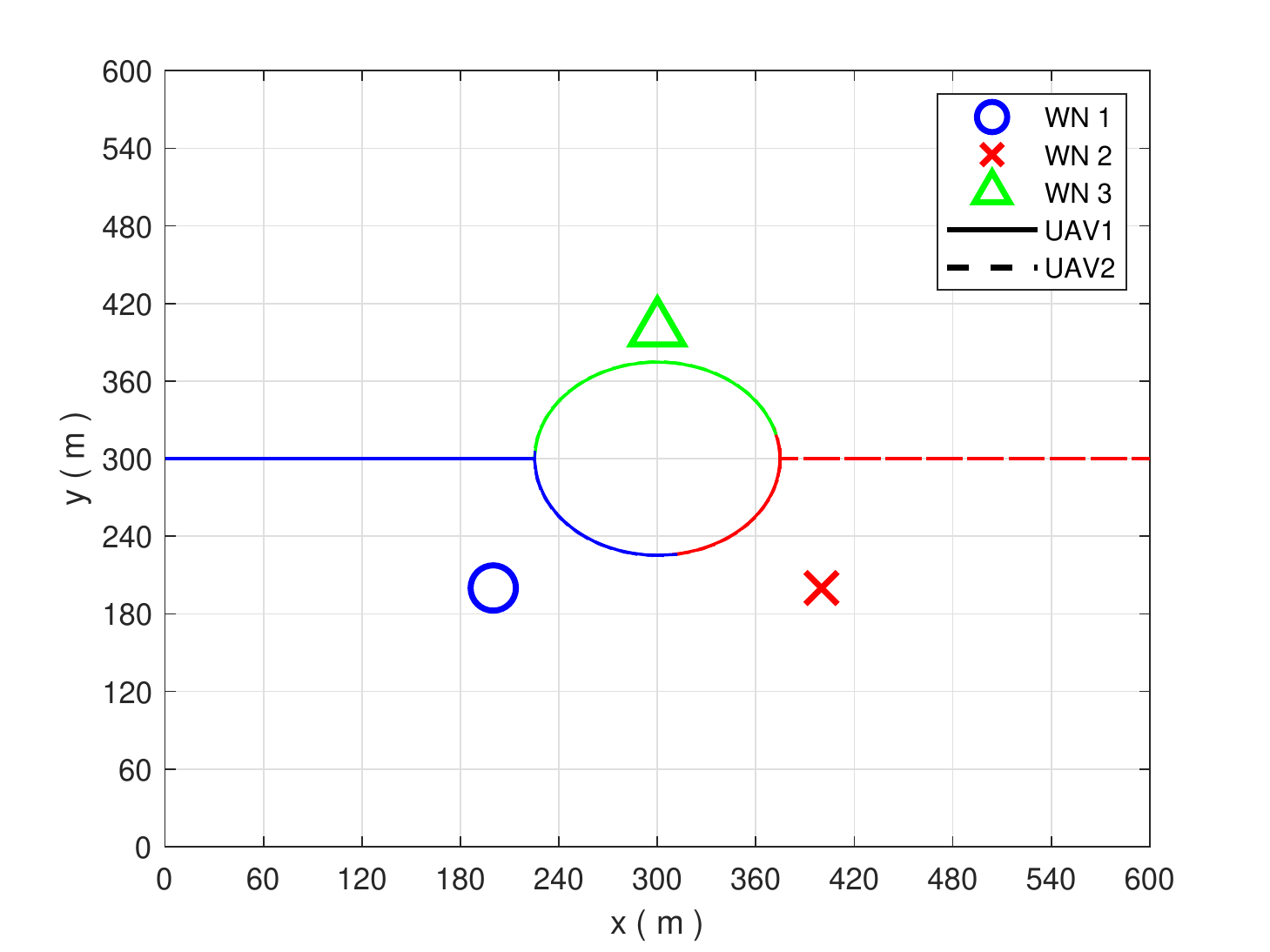}}
\caption{The UAV flight plans and the nearest node association for the three heuristic schemes.}
\label{fig_heuristic_fly}
\end{figure*}

\subsection{Performance of Offline Designs}
Fig. \ref{fig_coverge} shows the convergence of the proposed offline method for different numbers of WNs in the afternoon. The performance can be monotonically improved until convergence, which validates our analysis in Sec. \ref{Convergence_analysis}. Moreover, the required number of outer iterations for convergence increases with the number of WNs, and the performance is converged within four and eleven iterations for $K=2$ and $K=5$, respectively.
\begin{figure}[htbp]
\centering
\includegraphics[width=0.48\textwidth]{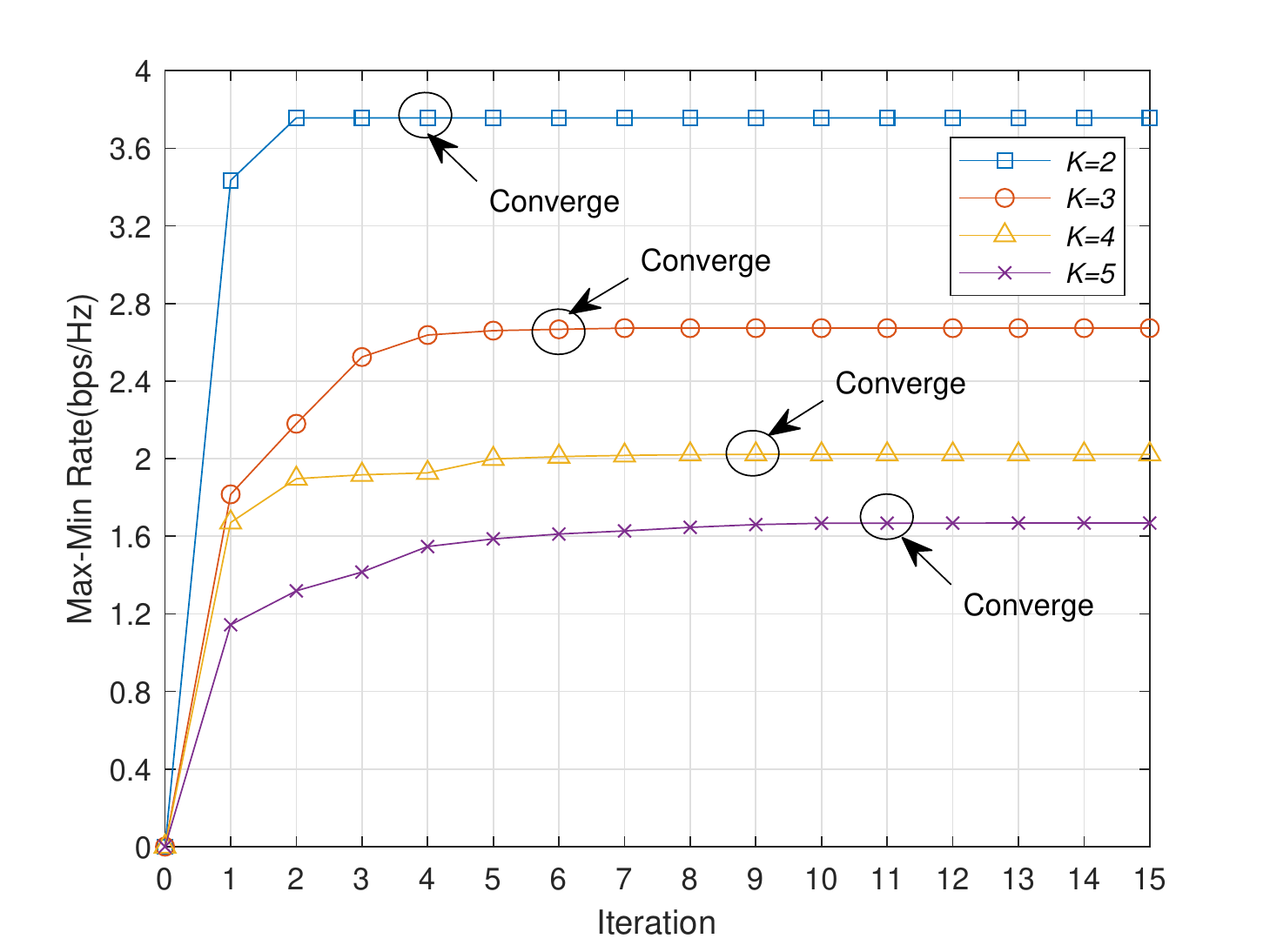}
\caption{Convergence of the proposed offline method for different numbers of WNs in the afternoon.}
\label{fig_coverge}
\end{figure}

\begin{figure}[htbp]
\centering
\includegraphics[width=0.48\textwidth]{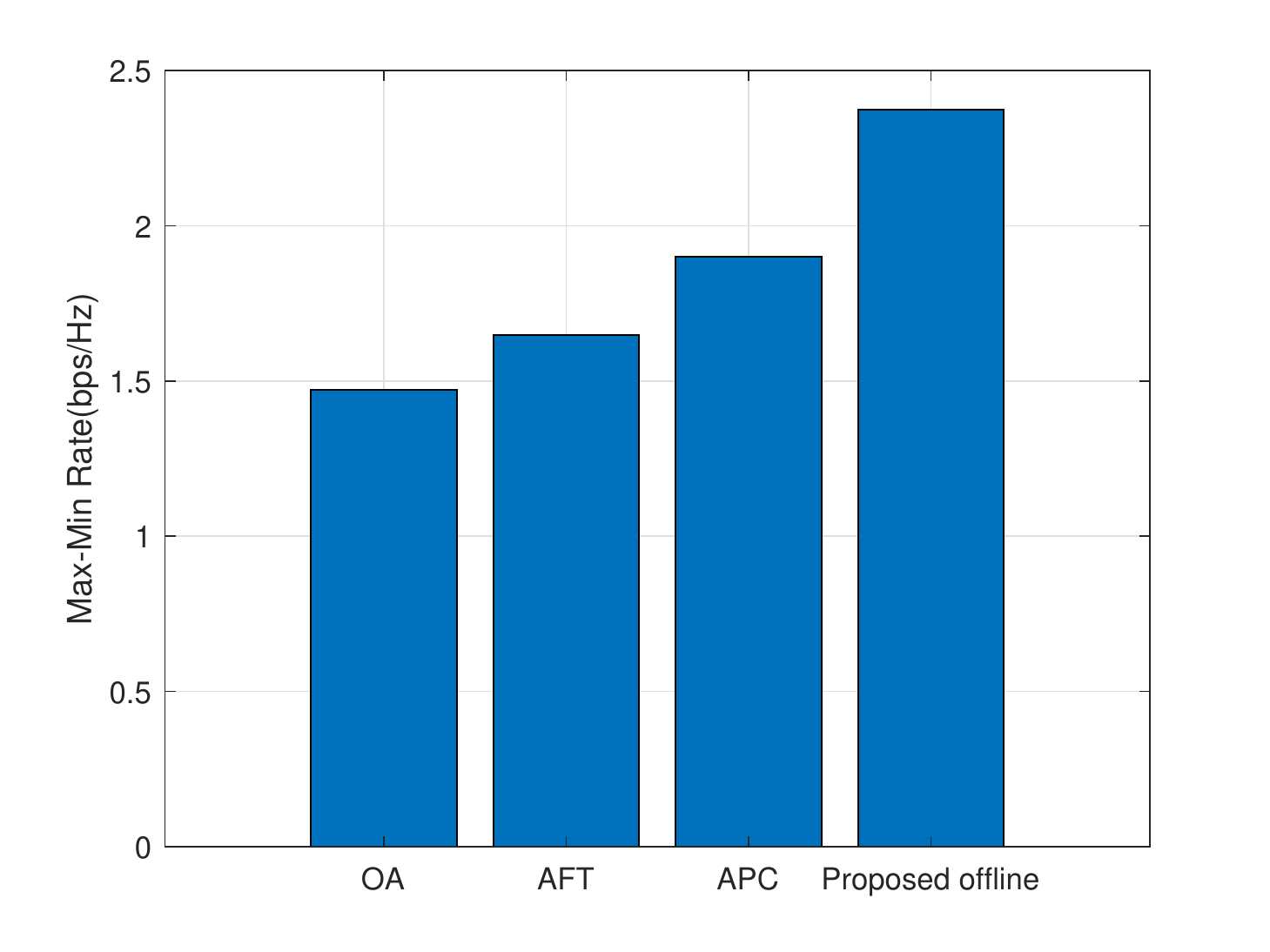}
\caption{Comparison of different combinations of design factors for offline optimization in the afternoon ($K=3$).}
\label{fig_apq}
\end{figure}

To evaluate the influence of different combinations of design factors on the performance, the following offline methods are compared in Fig. \ref{fig_apq} by only optimizing some of the design factors: (1) only communication association (OA), (2) only communication association and UAV flight trajectory (AFT), (3) only communication association and transmit power control (APC). In case that no optimization is applied, the exhaustive power control method and the UC flight trajectory in Fig. \ref{fig_heuristic_fly} are used. Here, the solar irradiance in the afternoon is adopted, and $K=3$. From this figure, the proposed offline method with fully joint optimization performs much better than the other methods that optimize only one or two design factors. Comparing AFT and APC shows that optimizing the transmit power can improve the performance more than optimizing the UAV trajectory. 

\begin{figure*}[htbp]
\centering 
\subfigure[{$K=3$ in the morning}]{
\label{fig.p1_fly}
\includegraphics[width=0.48\linewidth]{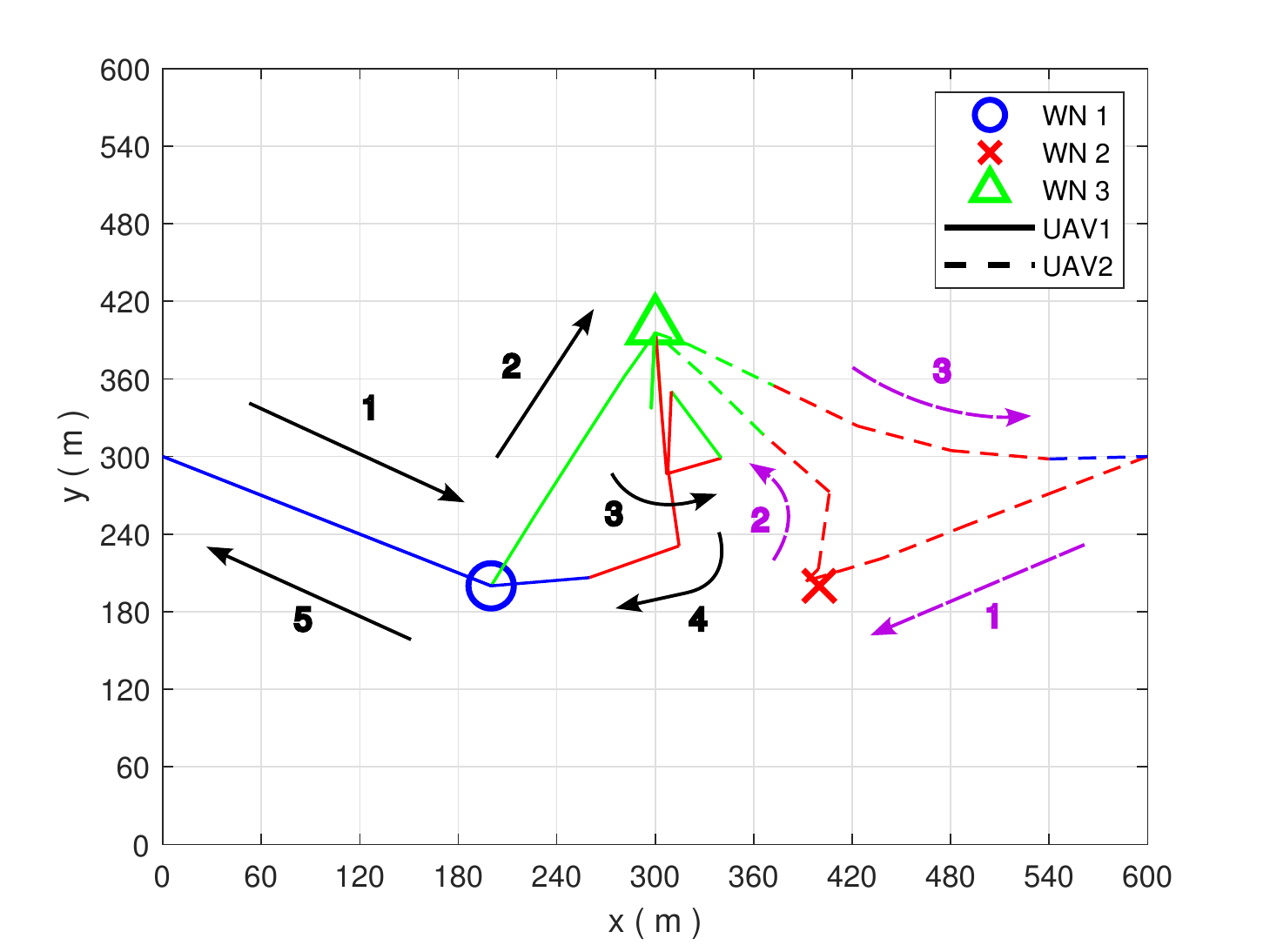}}
\subfigure[{$K=3$ in the afternoon}]{
\label{fig.p2_fly}
\includegraphics[width=0.48\linewidth]{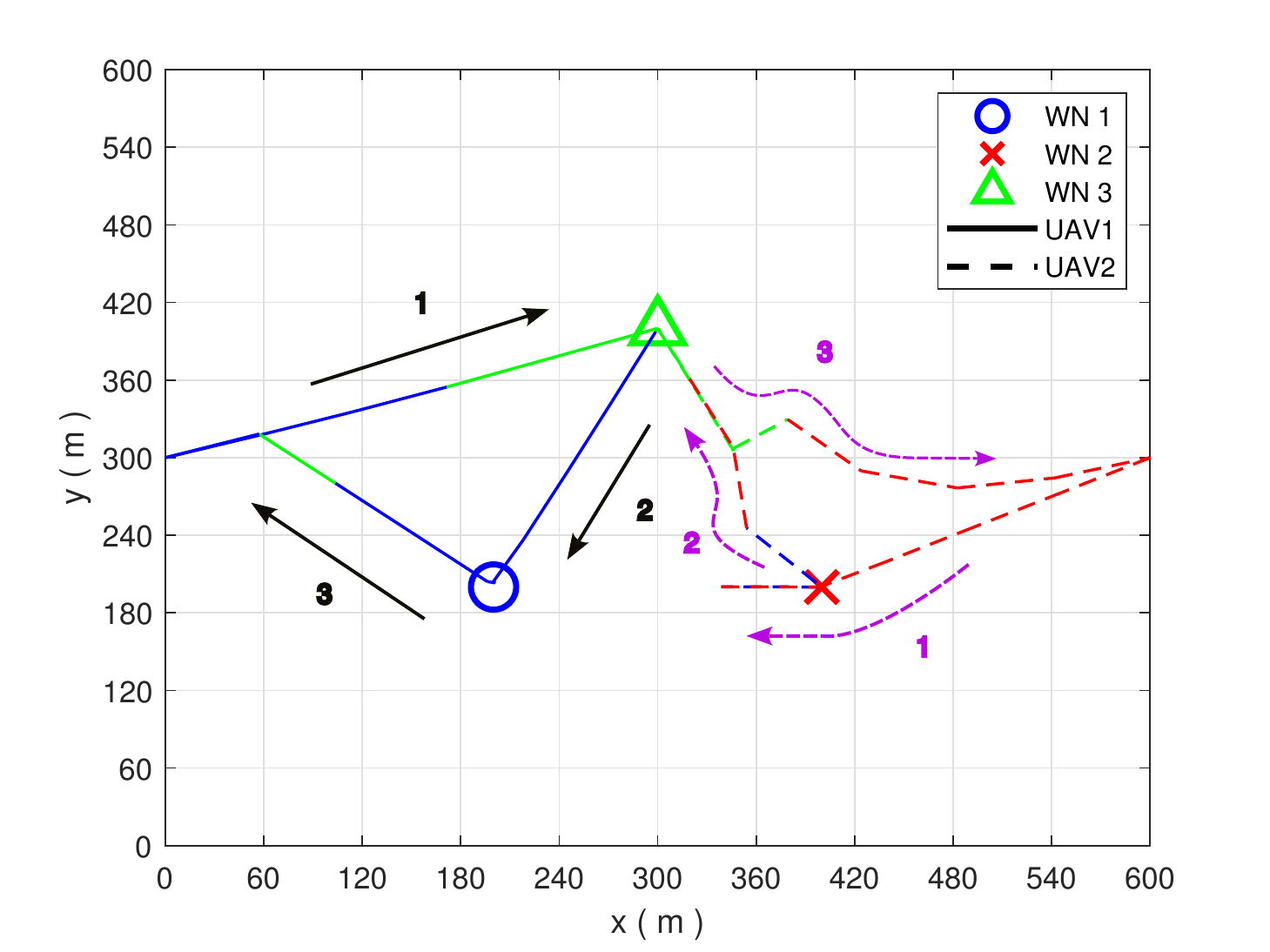}}
\caption{The optimal UAV flight trajectory and communication association of the proposed offline method in the morning and afternoon.}
\label{fig.p1p2_fly}
\end{figure*}

With the proposed offline method, Fig. \ref{fig.p1p2_fly} shows the optimal flight trajectory and communication association of the two UAVs during different EH periods when $K=3$. The colors of the UAV flight paths represent the user association results. We can find that in this deployment, WN 3 is the farthest node from the two UAVs, and both UAVs serve WN 3 in the morning and afternoon to improve the worst user rate. As can be seen, in the early stage of the mission, UAV 1 decides to serve WN 3 because of the abundant energy harvested in the afternoon. On the other hand, in the morning, UAV 1 serves its closest node WN 1 in the early stage while serving WN 3 in the middle stage, for which WN 3 collects enough battery energy and gets closer to UAV 1.

Fig. \ref{fig_245user_fly} shows the optimal UAV flight trajectory and communication association of the proposed offline method for different numbers of WNs. Obviously, the placement of WNs affects the optimal flight paths and communication association. For example, when $K=2$, the distances from the two WNs to the two UAVs are equal, and the two UAVs fly clockwise and counter-clockwise to avoid the interference problem. Similar observations can be found in the other cases. Taking another example of $K=4$, each UAV tends to serve two closer WNs.

\begin{figure*}[htbp]
\centering 
\subfigure[{$K=2$}]{
\label{Fig.sub.1}
\includegraphics[width=0.32\linewidth]{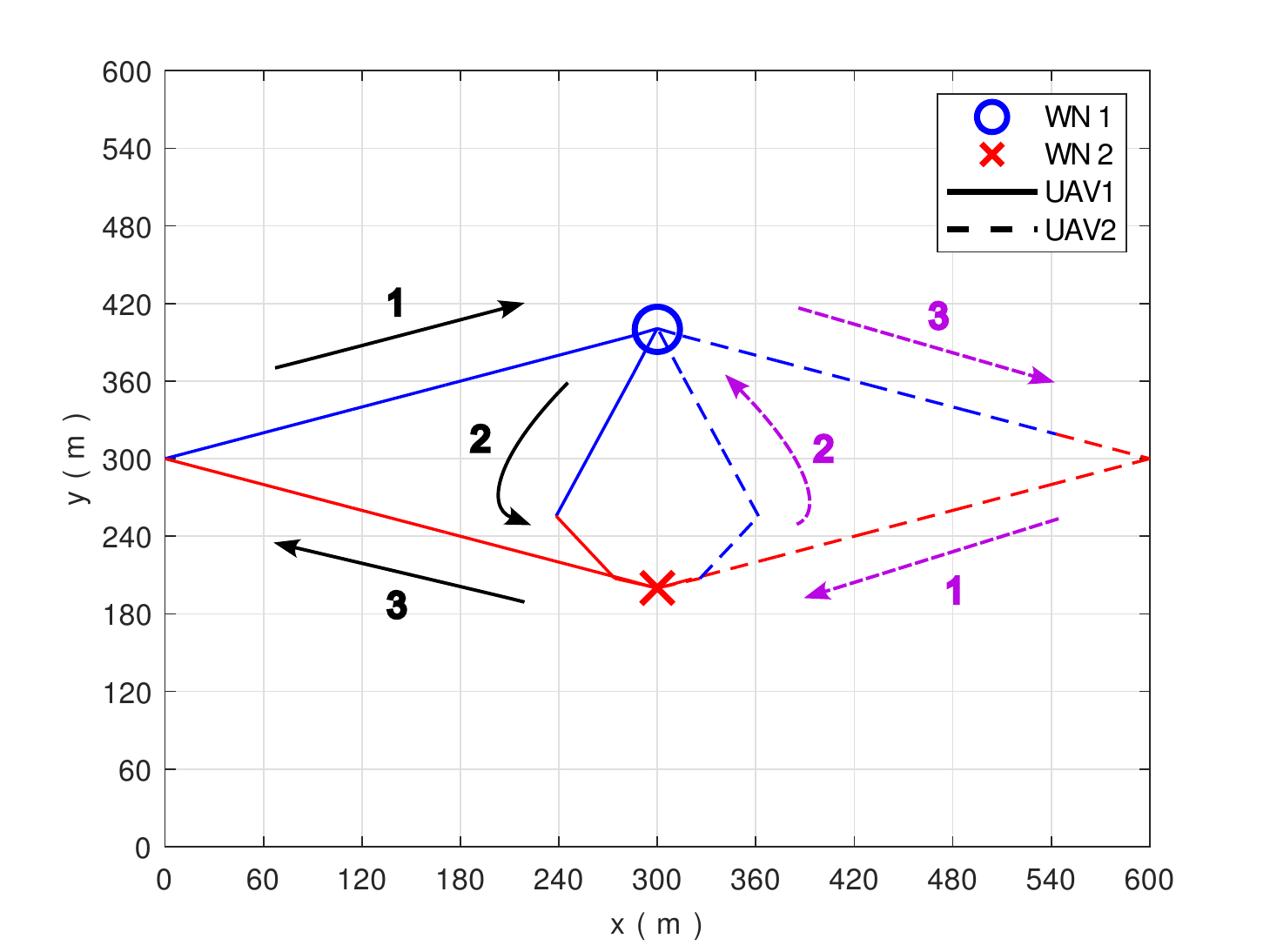}}
\subfigure[{$K=4$}]{
\label{Fig.sub.2}
\includegraphics[width=0.32\linewidth]{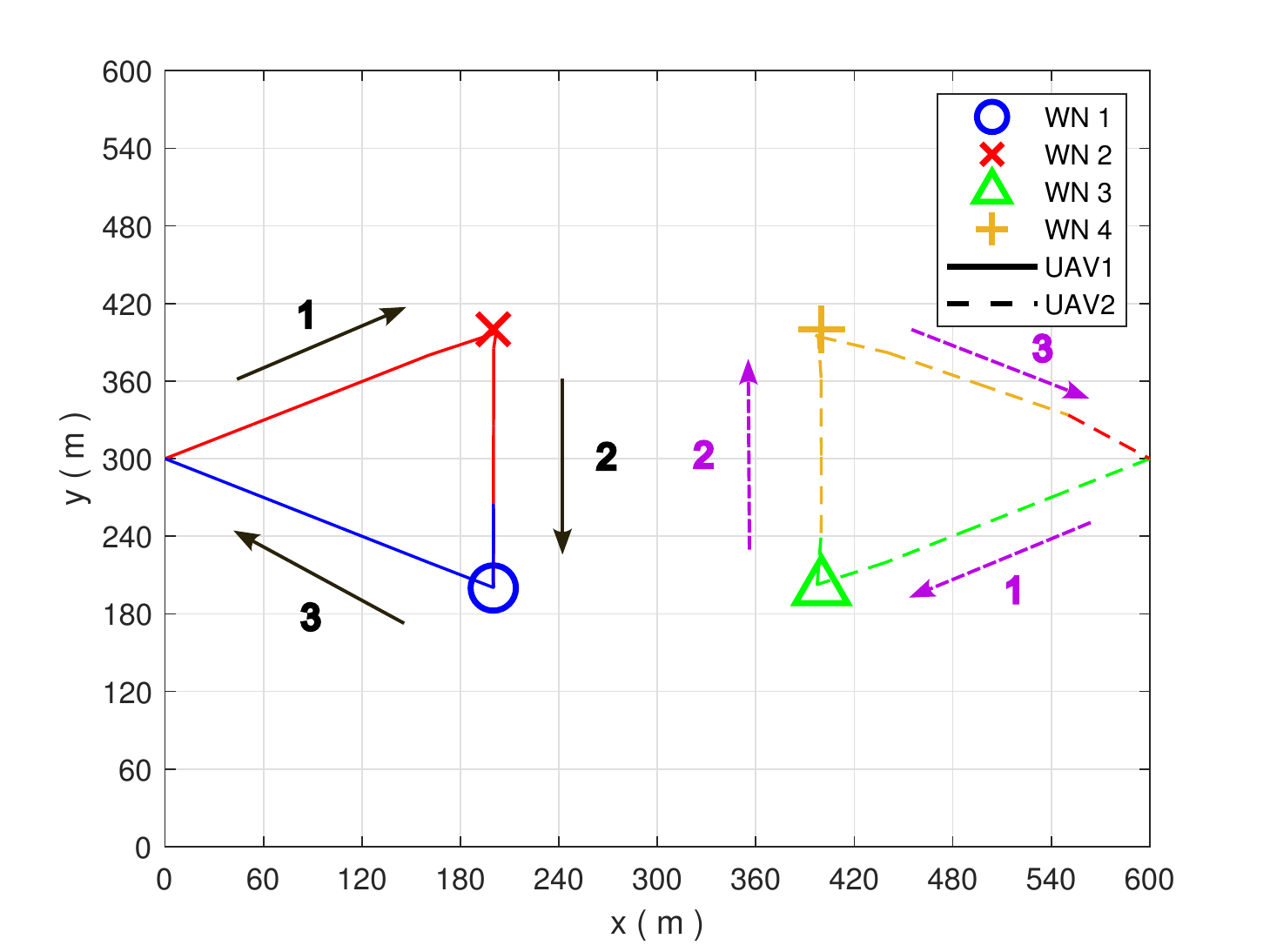}}
\subfigure[{$K=5$}]{
\label{Fig.sub.3}
\includegraphics[width=0.32\linewidth]{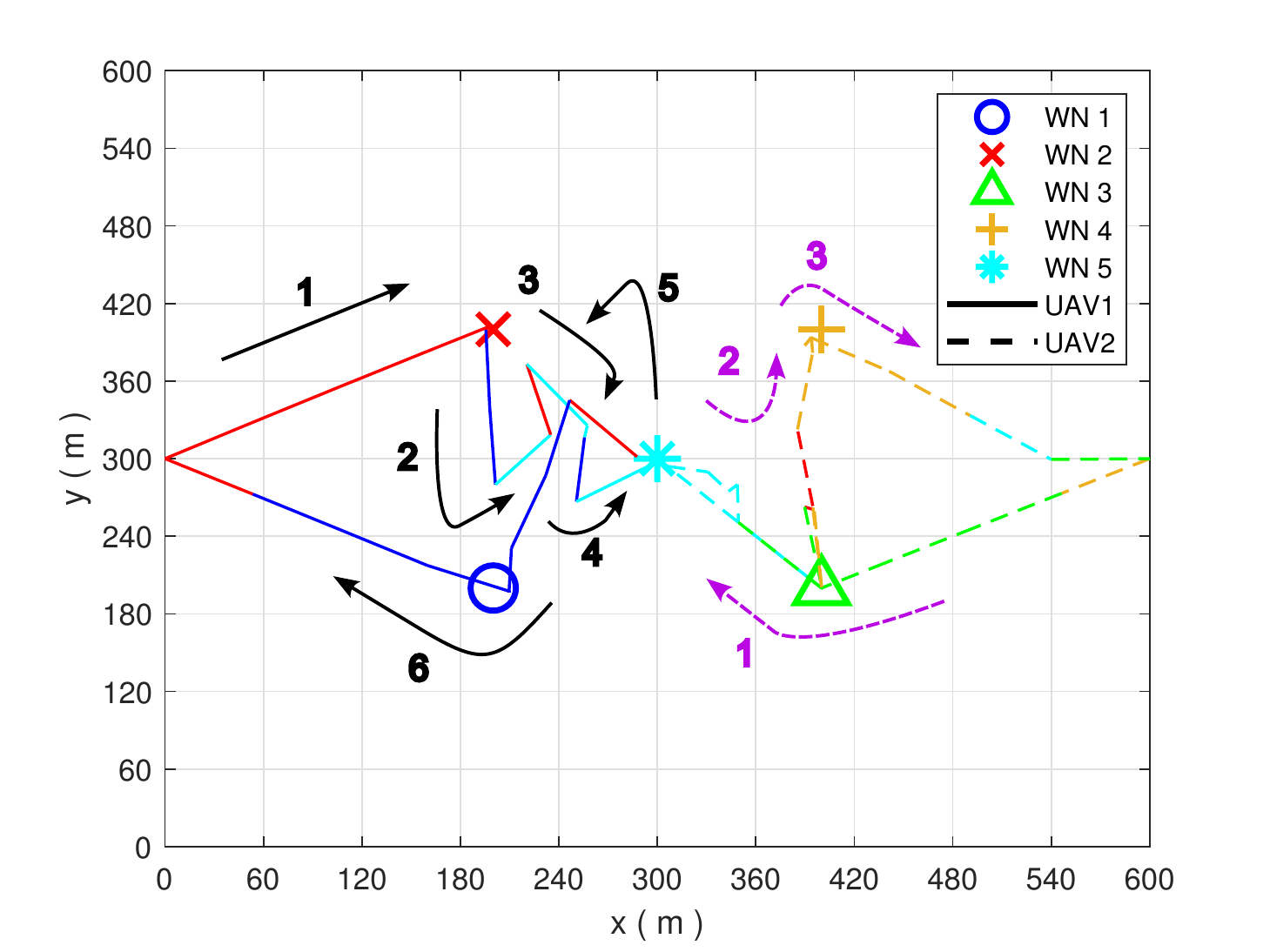}}
\caption{The optimal UAV flight trajectory and communication association of the proposed offline method in the afternoon for various numbers of WNs.}
\label{fig_245user_fly}
\end{figure*}

\begin{figure}[htbp]
\centering
\includegraphics[width=0.48\textwidth]{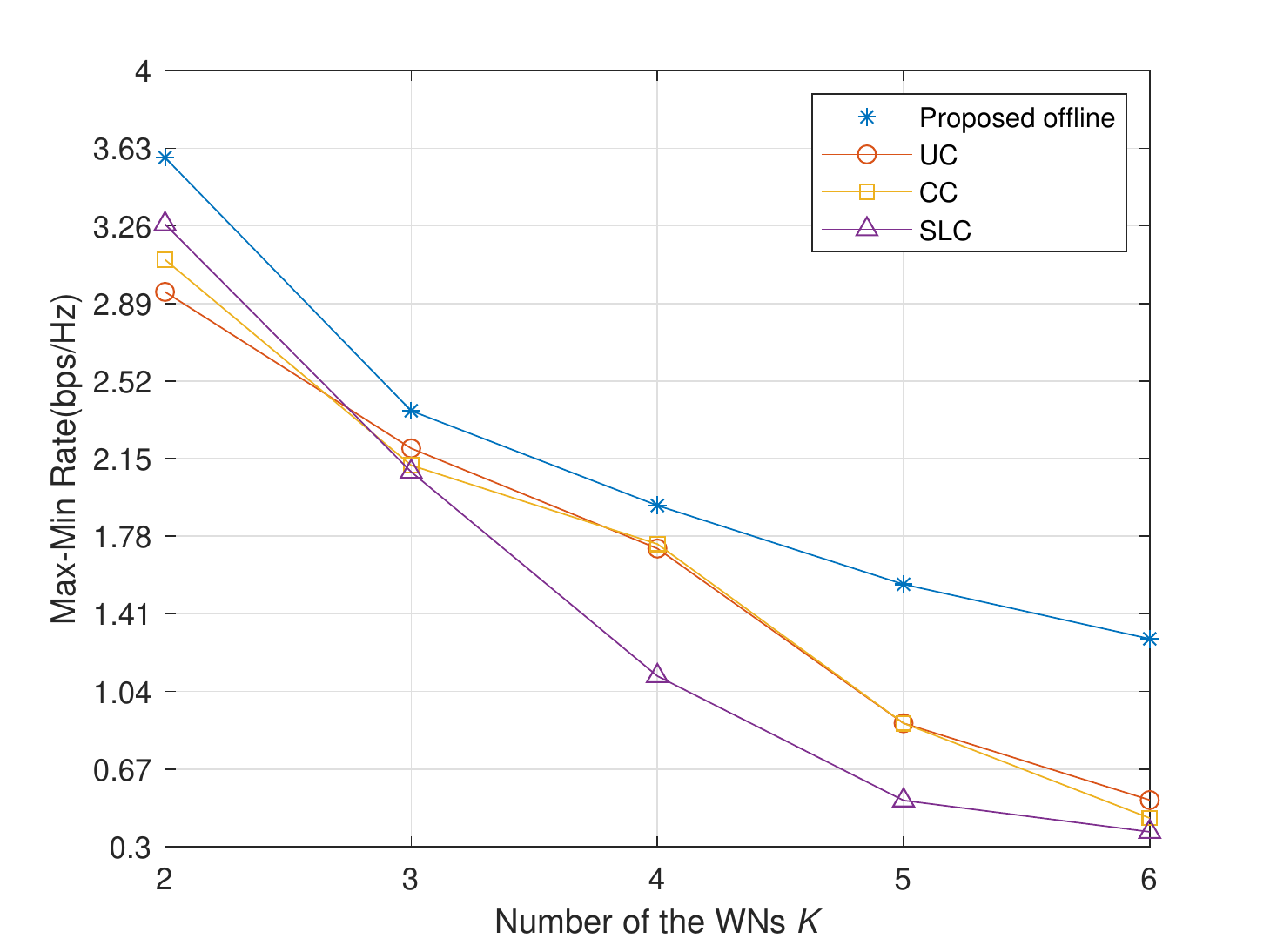}
\caption{Performance comparison of three heuristic schemes and the proposed offline method for various numbers of WNs in the afternoon.}
\label{fig_heuristic}
\end{figure}

Fig. \ref{fig_heuristic} compares three heuristic methods with the proposed offline method for different numbers of WNs\footnote{The positions of WNs for $K=6$ are $\textbf g_1=[200,200]^T$, $\textbf g_2=[200,400]^T$, $\textbf g_3=[400,200]^T$, $\textbf g_4=[400,400]^T$, $\textbf g_5=[200,300]^T$ and $\textbf g_6=[300,300]^T$.}. For the three heuristic methods, the nearest user association and exhaustive power control methods are applied with the three different flight plans in Fig. \ref{fig_heuristic_fly}. As can be seen, the proposed offline method is superior to the three compared methods, and the performance gap expands with the increase of $K$.

Fig. \ref{fig_battery_compare} shows the performance of the proposed offline method in the afternoon under different battery capacities. A larger battery capacity can improve the worst user rate performance. For a larger number of WNs, this improvement becomes more significant as each WN has less time to access the channel and is allowed to store more energy when waiting for data transmission. 

\begin{figure}[htbp]
\centering
\includegraphics[width=0.48\textwidth]{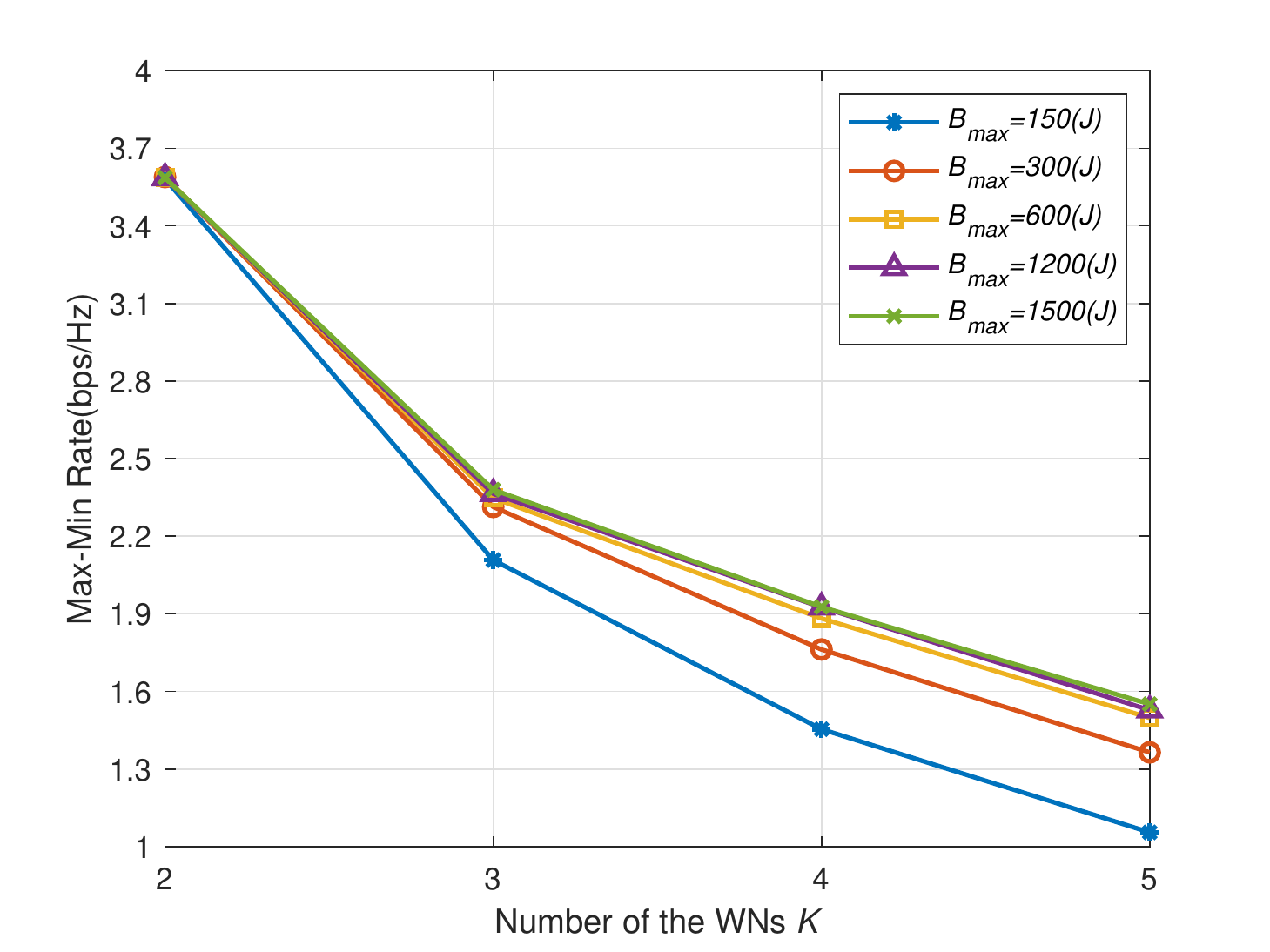}
\caption{Performance of the proposed offline method in the afternoon under different battery capacities $B_{\max}$.}
\label{fig_battery_compare}
\end{figure}
\subsection{Performance of Online Designs}
Fig. \ref{fig_reward} shows the performance of the CARL method with the three different reward designs under various noise power values. It is clearly seen that the ISR method can achieve the best performance, in terms of the worst user rate and the probability of successful mission, i.e., all the UAVs successfully fly back to the initial point, among the thee reward designs. This implicitly suggests that it is more appropriate to adopt the instantaneous average sum rate of all WNs, rather than the worst accumulated sum rate, as a reward in the online learning for balancing user rates. Therefore, the ISR reward is used for the CARL and the conventional RL (i.e., the reward in (\ref{Convential_Reward}) is set as $f\left( R_{k,n}\right)=\frac{1}{K}\sum_{k=1}^{K}\nolimits{R_{k,n}}$) in the following simulations.

\begin{figure}[htbp]
\centering
\includegraphics[width=0.48\textwidth]{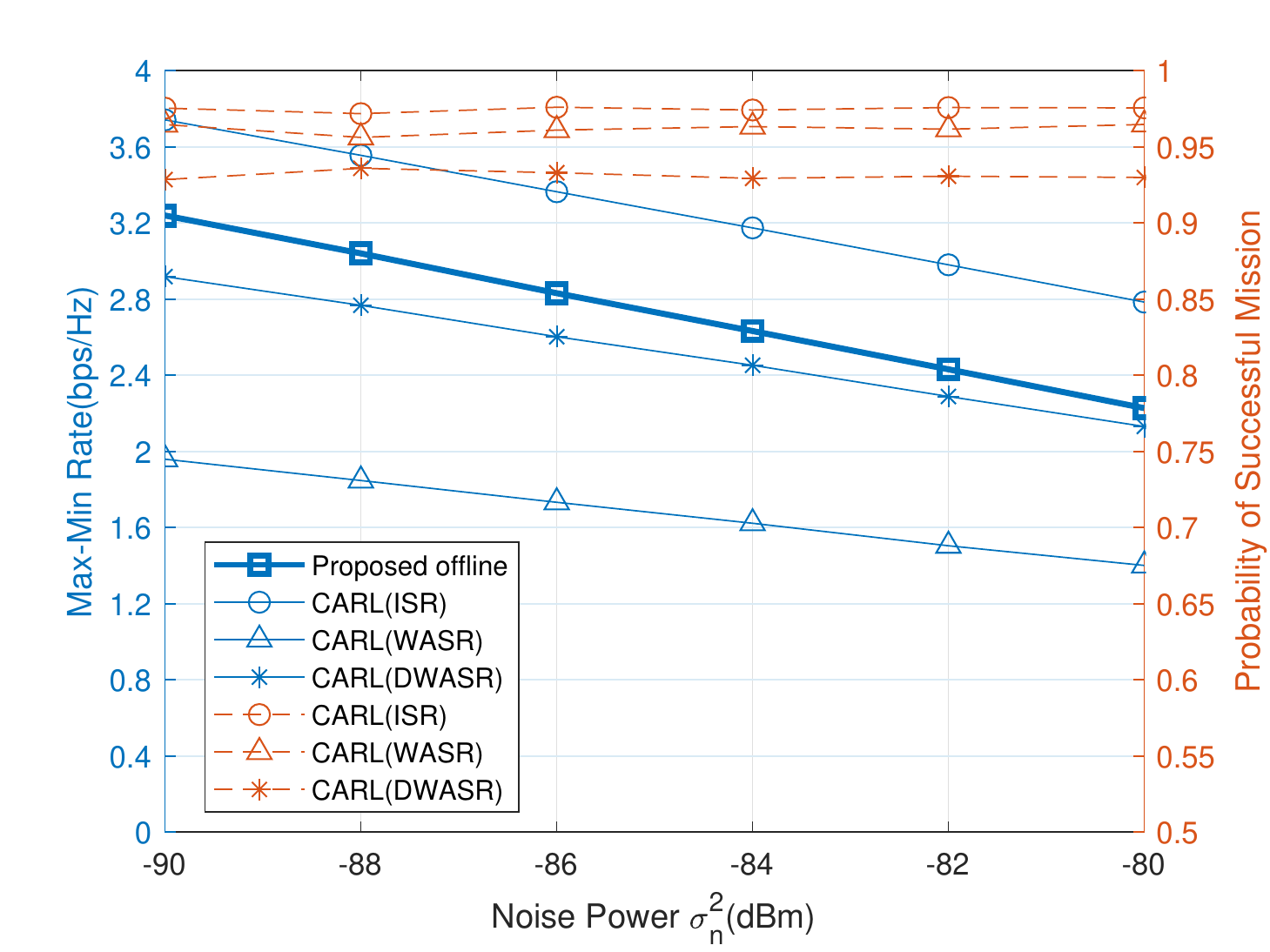}
\caption{Performance of the CARL method with the three different reward designs in the morning ($K=3$).}
\label{fig_reward}
\end{figure}

In Fig. \ref{fig_C_CARL_CVX}, we compare the max-min user rates of the proposed offline and CARL methods for $K=3$ under various noise power values in the afternoon. The performance of the CARL method that utilizes a non-optimal UC trajectory is also simulated for comparison. It shows that the CARL can further enhance the performance of the proposed offline method and outperform the CARL method with a non-optimal UC trajectory. Regarding the probability of successful mission, the CARL with the optimal offline trajectory can attain a probability of $0.935$, which is slightly better than the CARL with the UC trajectory. 
\begin{figure}[htbp]
\centering
\includegraphics[width=0.48\textwidth]{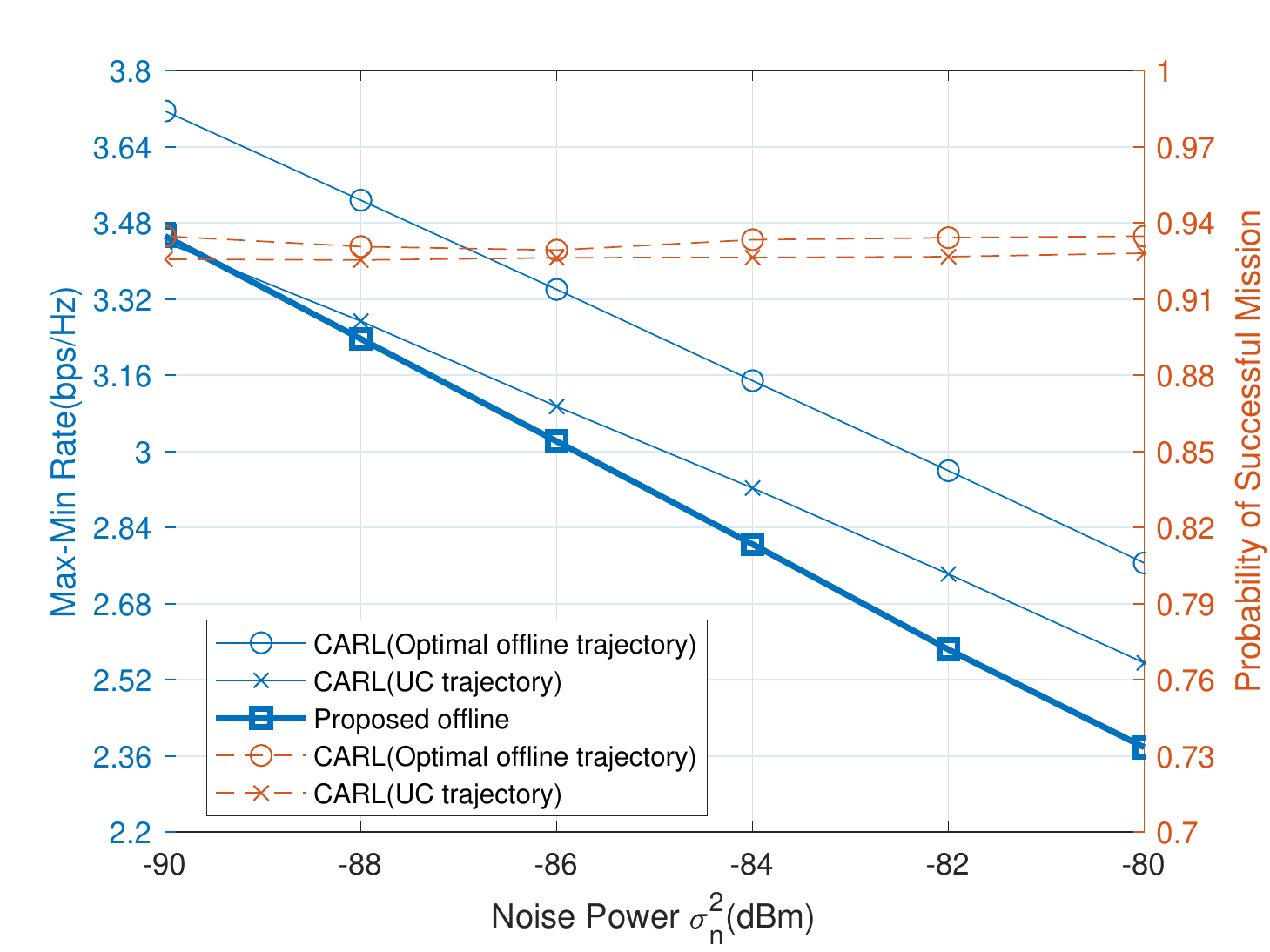}
\caption{Performance of the CARL, conventional RL, and offline methods under various noise power values in the afternoon ($K=3$).}
\label{fig_C_CARL_CVX}
\end{figure}

Fig. \ref{fig_DF} shows the performance of the CARL method for different widths of the flight corridor $D_F$ in the morning. We can see that the flight corridor width performs best at $D_F=60$ m when $N_e$ is small, while the flight corridor width that achieves the best performance becomes $D_F=90$ m when $N_e$ increases. This is because the number of exploration states is fewer when $D_F$ is small, and the Q-table converges quickly even if $N_e$ is small. Moreover, as $N_e$ increases, the performance improvement for $D_F=60$ m saturates, whereas the other larger flight corridor widths can still offer continuous improvement. Hence, a larger flight corridor width can potentially improve the performance at the expense of more training epochs $N_e$.
\begin{figure}[htbp]
\centering
\includegraphics[width=0.48\textwidth]{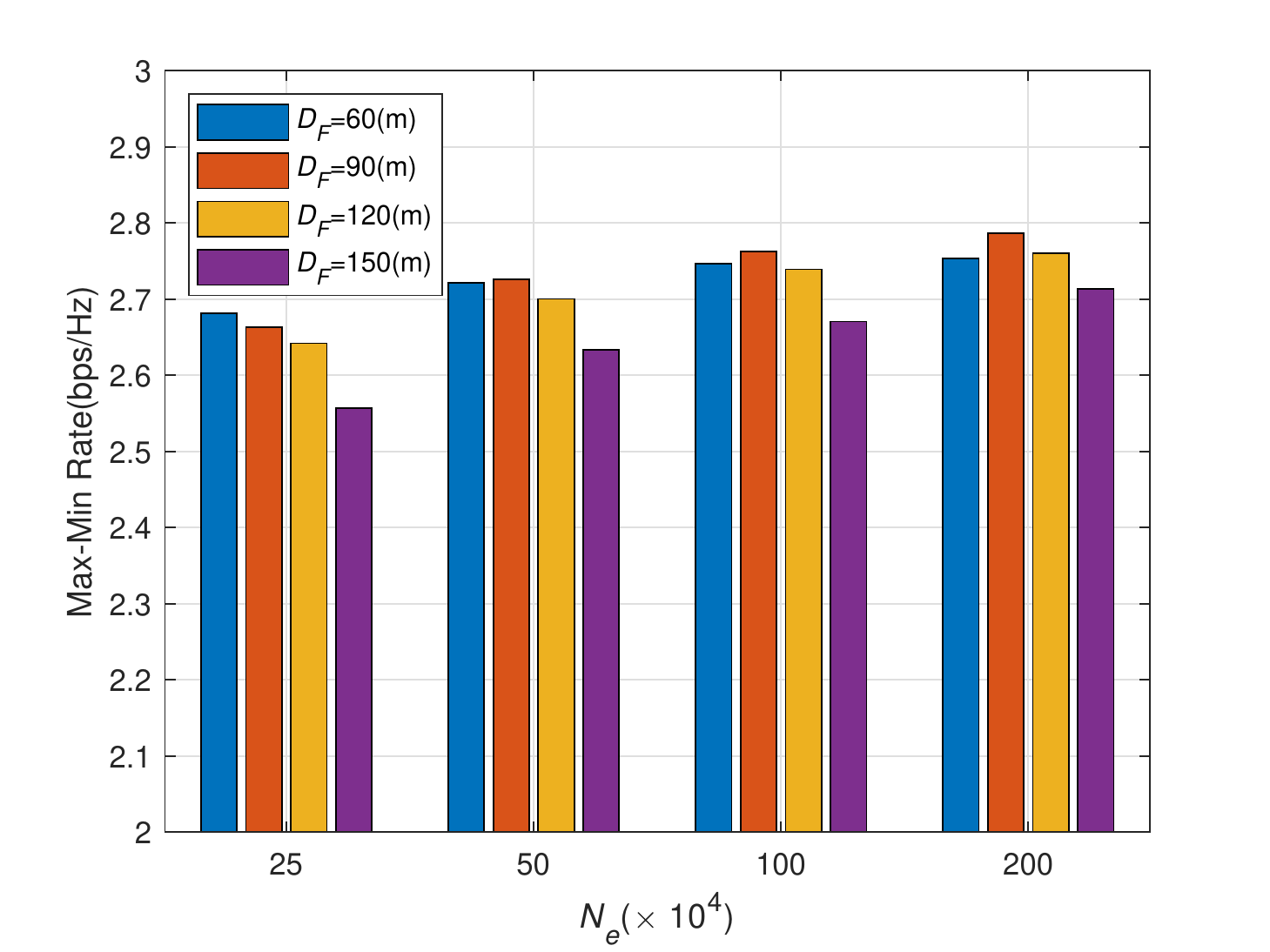}
\caption{Performance of the CARL method under different widths of the flight corridor $D_F$ in the morning ($K=3$).}
\label{fig_DF}
\end{figure}

In Fig. \ref{fig_CRLvsCVX}, we compare the performance of the CARL, conventional RL and proposed offline methods under different training epochs $N_e$. We can see that the max-min rate of the CARL method increases with $N_e$, whereas there is no significant performance improvement for the conventional RL method. This is because the performance of the conventional RL saturates quickly, while the performance of the CARL can be efficiently improved under the guidance of optimal offline trajectories. In addition, the probability of successful mission of the conventional RL is higher than that of the CARL when $N_e$ is small, but the probability of the CARL eventually surpasses that of the conventional RL as $N_e$ exceeds $50\times 10^4$. This is because the reward design in the conventional RL forces the UAVs to fly back to the initial point with less exploration, while the CRAL can learn to fly back to the initial point by following the flight corridor if the training number is sufficiently large. Also, both online methods outperform the offline method due to the dynamic responses to the changes of channel and battery conditions. 

\begin{figure}[htbp]
\centering
\includegraphics[width=0.48\textwidth]{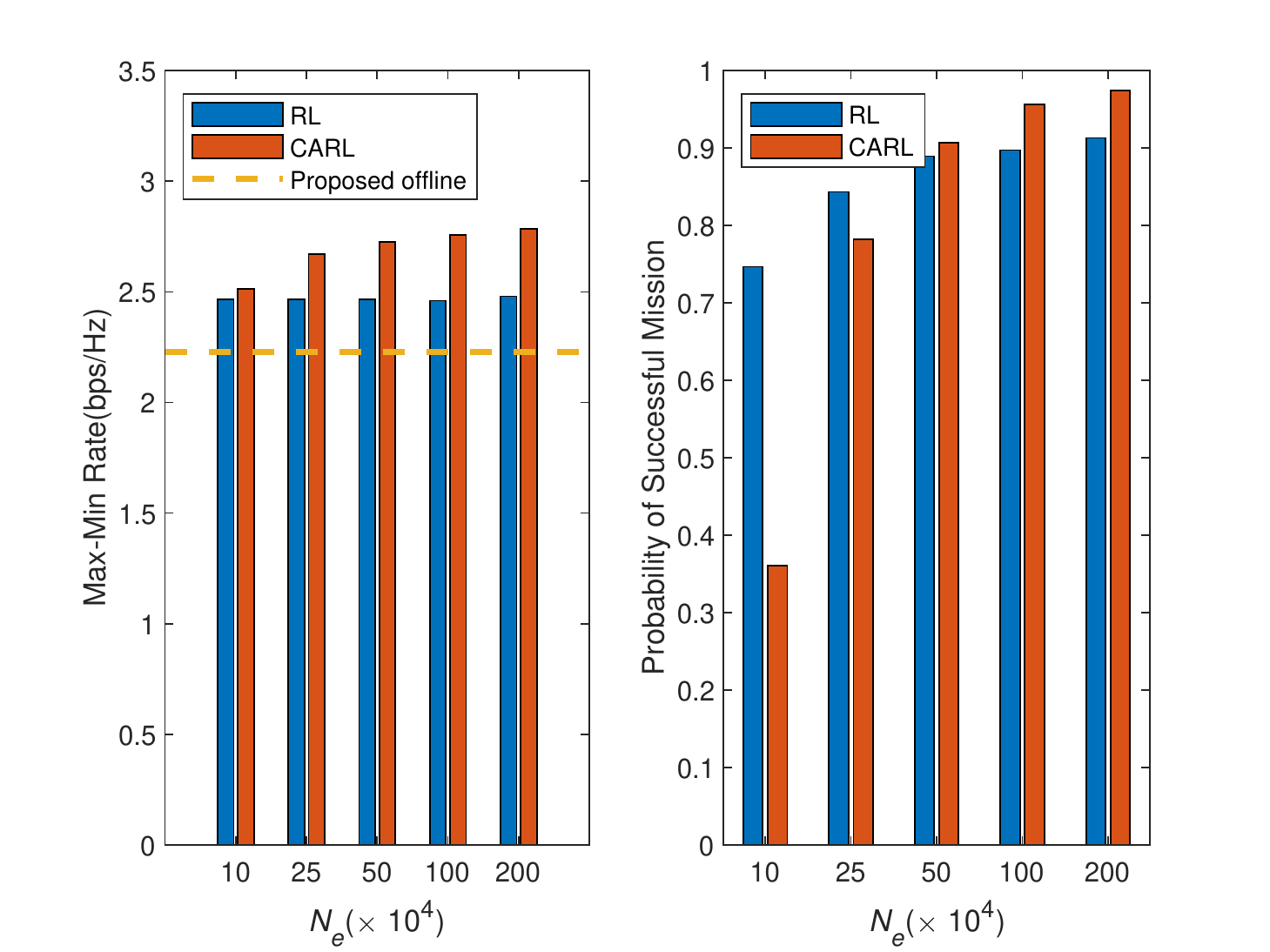}
\caption{Performance of the CARL, conventional RL, and offline methods under various numbers of training periods $N_e$ in the morning ($K=3$).}
\label{fig_CRLvsCVX}
\end{figure}

\section{Conclusion}\label{Conclusion}
In this paper, we investigate the joint design problem of multi-UAV flight trajectories, communication association between UAVs and ground nodes, and uplink power control for a multi-UAV network with multiple uplink EH nodes. Under a mixed channel model with LOS/NLOS and small-scale fading, a series of SCAs are performed to deal with this non-convex joint design problem, and an offline method that does not require causal knowledge of instantaneous CSI and ESI is proposed. An online CARL method is proposed to further improve the performance by exploiting preset UAV flight corridors according to the optimal offline UAV trajectories. Through efficient flight guidance and reducing the number of state spaces to be explored, the CARL can improve the learning convergence and performance compared to the conventional RL without the assistance from the offline UAV trajectories.

\appendices
\section{Proof of Lemma \ref{log_lower_bound__convex}}\label{Lemma_1}
We first consider a function $\phi(x,y)=\ln(\frac{C_1C_2}{xy}+\frac{C_3}{x})$ with two variables $x>0$ and $y>0$, where $C_1, C_2, C_3>0$. The Hessian matrix of $\phi(x,y)$ is
\setcounter{equation}{0}
\renewcommand\theequation{A.\arabic{equation}}
\begin{align}
  \nabla^2\phi(x,y)=
  \begin{bmatrix} \frac{\partial^2\phi(x,y)}{\partial x^2}&\frac{\partial^2\phi(x,y)}{\partial x\partial y} \\
                  \frac{\partial^2\phi(x,y)}{\partial x\partial y}&\frac{\partial^2\phi(x,y)}{\partial y^2} \\
  \end{bmatrix},
\end{align}
where each term can be derived as
\begin{align}
    &\frac{\partial^2\phi(x,y)}{\partial x^2}=\frac{1}{x^2}; \\
    &\frac{\partial^2\phi(x,y)}{\partial y^2}=\frac{C_1C_2(C_1C_2+2C_3y)}{(C_1C_2y+C_3y^2)^2};\\
    &\frac{\partial^2\phi(x,y)}{\partial x\partial y}=\frac{\partial^2\phi(x,y)}{\partial y\partial x}=0.
\end{align}
Since $\mathsf{det}(\nabla^2\phi(x,y))\geq0$ and $\mathsf{tr}(\nabla^2\phi(x,y))\geq0$, the Hessian matrix $\nabla^2\phi(x,y)$ is positive semi-definite. Thus, the function $\phi(x,y)$ is convex in $(x,y)$.

From (\ref{average_H}), (\ref{eq35add}) and (\ref{eq36add}), we rewrite $\check R_{1m}[n]$ as a log-sum-exp form:
\begin{align}
&\check R_{1m}[n]= {\log}_2\left(\sum_{i=1}^{K} e^{ {\tilde{\phi}}_{i}(X_{m,i}[n] ,Y_{m,i}[n] ) } +\sigma_n^2\right) ,
\end{align}
where $ {\tilde{\phi}}_{i}(X_{m,i}[n] ,Y_{m,i}[n] ) = \ln \big(\frac{P_i[n]C_1 C_2}{X_{m,i}[n] Y_{m,i}[n]}+\frac{P_i[n]C_3}{X_{m,i}[n]} \big)$. By using the fact that a function ${\log}_2\left(\sum_{i=1}^{K} e^{g_i(x)}  \right)$ is convex whenever $g_i(x)$ is convex for all $i$ {\cite{S. Boyd04}}, it can be shown that ${\log}_2\left(\sum_{i=1}^{K} \exp {\big\{  \tilde{\phi}_{i}(X_{m,i}[n] ,Y_{m,i}[n]) \big\} } \right)$ is convex in $X_{m,i}[n]$ and $Y_{m,i}[n]$ for all $i$, since $\tilde{\phi}_{i}(X_{m,i}[n] ,Y_{m,i}[n] ) $ is convex in $(X_{m,i}[n] ,Y_{m,i}[n])$. As a result, $\check R_{1m}[n]$ is convex in $X_{m,i}[n]$ and $Y_{m,i}[n]$ for all $i$.

\section{Proof of Theorem \ref{first_order_R_lowerbound}}\label{first_order_R_lowerbound_proof}
We first consider a convex function $\varphi({\textbf x}, {\textbf y})$ with two variables $\textbf x$ and $\textbf y$, and a lower bound can be obtained by using the first-order Taylor expansion at $\left({\textbf x},{\textbf y}\right)=\left({\textbf x}_0, {\textbf y}_0\right)$:
\setcounter{equation}{0}
\renewcommand\theequation{B.\arabic{equation}}
\begin{align}
    &\varphi({\textbf x},{\textbf y}) \geq\varphi({\textbf x}_0,{\textbf y}_0)+ \nabla_{\textbf x} \varphi({\textbf x},{\textbf y})   \Big|_{({\textbf x},{\textbf y})=({\textbf x}_0,{\textbf y}_0)} ({\textbf x}- {\textbf x}_0) \nonumber\\
    &\;\;\;\;\;\;\;\;\;\;\;\;\;\;\;\;\;+  \nabla_{\textbf y} \varphi({\textbf x},{\textbf y}) \Big|_{({\textbf x},{\textbf y})=({\textbf x}_0,{\textbf y}_0)} ({\textbf y}-{\textbf y}_0). \label{B1}
\end{align}
For a given $\textbf q_m [n]=\textbf q_m^r [n]$ and from (\ref{eq35add}) and (\ref{eq36add}), we can calculate $X^r_{m,i}[n]$ and $Y^r_{m,i}[n]$ as in (\ref{eq41_1}) and (\ref{eq41_2}). From Lemma \ref{log_lower_bound__convex}, the function $\check R_{1m}[n]={\log}_2\big(\sum_{i=1}^{K}P_i[n]\big(\frac{C_1C_2}{X_{m,i}[n]Y_{m,i}[n]}+\frac{C_3}{X_{m,i}[n]}\big)+\sigma_n^2\big)$ is a convex function with respect to the variables $X_{m,i}[n]$ and $Y_{m,i}[n]$ for all $i$. By using (\ref{B1}), we can derive the first-order Taylor expansion of $\check R_{1m}[n]$ at $\big(X_{m,i}[n], Y_{m,i}[n]\big)= \big( $ $X_{m,i}^{r}[n], Y_{m,i}^{r}[n]\big)$, given by $\check R_{1m}^{lb}[n]$ in (\ref{eq41}). Hence, we can get $\check R_{1m}[n] \geq \check R_{1m}^{lb}[n]$.

\section{Proof of Theorem \ref{theorem_Y_upperbound}} \label{upper_bound_Y_proof}
Define a variable $U_{m,i}[n]=\|{\textbf q}_m[n]-{\textbf g}_i\|^2_2$. By applying the change of variables into (\ref{eq36add}), we can get
\setcounter{equation}{0}
\renewcommand\theequation{C.\arabic{equation}}
\begin{align}
    Y_{m,i}[n]=1+A e^{-B\left(\frac{180}{\pi}\tan^{-1}\left(\frac{H}{\sqrt{U_{m,i}[n]}}\right)-A\right)}.\label{C1}
\end{align}
Since $\tan^{-1}\left(\frac{1}{\sqrt{x}}\right)$ is a convex function when $x > 0$, the first-order Taylor expansion of $\tan^{-1}\left(\frac{H}{\sqrt{U_{m,i}[n]}}\right)$ at the point $U_{m,i}[n]=\|\textbf q^r_m[n]-\textbf g_i\|^2_2\triangleq U^r_{m,i}[n]$ can be derived as
\begin{align}
    &\tan^{-1}\left(\frac{H}{\sqrt{U_{m,i}[n]}}\right) \geq\tan^{-1}\left(\frac{H}{\sqrt{U^r_{m,i}[n]}}\right)\nonumber\\
    &\;\;-\frac{H}{2\sqrt{U_{m,i}^r[n]}\left(H^2+U^r_{m,i}[n]\right)} \times\left(U_{m,i}[n]-U^r_{m,i}[n]\right).\label{C2}
\end{align}
By simply applying (\ref{C2}) in (\ref{C1}), an upper bound $Y_{m,i}^{up}[n]$ can be obtained for $Y_{m,i}[n]$, as shown in (\ref{Y_upperbound}). In addition, it reveals in (\ref{Y_upperbound}) that $Y_{m,i}^{up}[n]$ is convex in ${\textbf q}_m[n]$, since the term $\left\|\textbf q_m[n]-\textbf g_i\right\|_2^2$ is a square norm function of $\textbf q_m[n]$.

\end{document}